%% file: main.tex
\documentclass[11pt,letterpaper]{article}

\usepackage[margin=1in]{geometry}
\usepackage{latexsym,graphicx,amssymb}
\usepackage{amsthm}
\usepackage{amsmath,enumerate}
\usepackage{wrapfig}
\usepackage{subfigure}
\usepackage{float}
\usepackage{psfig}
\usepackage{epsfig}
\usepackage{xspace}
\usepackage{paralist}
\usepackage{enumerate}
\usepackage{cases}
\usepackage{caption}
\usepackage{algorithm}
\usepackage{algorithmicx}
\usepackage[noend]{algpseudocode}
\usepackage{color}
\usepackage{complexity}

\usepackage{times}

\usepackage{multicol}
\newenvironment{proofof}[1]{\noindent{\bf Proof of #1:  }}{\hfill\rule{2mm}{2mm}}
\numberwithin{figure}{section}
\numberwithin{equation}{section}

\newtheorem{theorem}{Theorem}[section]

\newtheorem{corollary}[theorem]{Corollary}

\newtheorem{lemma}[theorem]{Lemma}
\newtheorem{claim}[theorem]{Claim}

\newtheorem{proposition}[theorem]{Proposition}

{
\theoremstyle{definition}
\newtheorem{definition}[theorem]{Definition}

}

\newcommand{\ignore}[1]{}

\newcommand{\plus}[1]{\ensuremath{\left[ #1 \right]^+}}

\newcommand{\expct}{\ensuremath{\text{{\bf E}}}}

\newcommand{\Te}{\mathsf{T}_e}
\newcommand{\He}{\mathsf{H}_e}
\newcommand{\tail}{\mathsf{T}}
\newcommand{\head}{\mathsf{H}}

\newcommand{\T}{\ensuremath{\mathsf{T}}\xspace}
\newcommand{\norm}[1]{\left\lVert#1\right\rVert}
\newcommand{\I}{\mathsf{I}}
\newcommand{\Q}{\mathsf{Q}}

\newcommand{\Lc}{\mathcal{L}}

\newcommand{\vp}{\varphi}

\newcommand{\ra}{\rightarrow}

\renewcommand{\R}{\mathbb{R}}

\newcommand{\xcite}[1]{\cite{#1}}
\newcommand{\STATE}{\State}
\newcommand{\FOR}{\For}
\newcommand{\WHILE}{\While}
\newcommand{\ENDWHILE}{\EndWhile}
\newcommand{\ENDFOR}{\EndFor}
\newcommand{\RETURN}{\Return}

\newcommand*\samethanks[1][\value{footnote}]{\footnotemark[#1]}

\title{Diffusion Operator and Spectral Analysis for \\ 
Directed Hypergraph Laplacian}

\author{T-H. Hubert Chan\thanks{Department of Computer Science, the University of Hong Kong. {\texttt{\{hubert,zhtang,xwwu,czzhang\}@cs.hku.hk}}} \and Zhihao Gavin Tang\samethanks
\and Xiaowei Wu\samethanks
 \and Chenzi Zhang\samethanks}
\date{}

\begin{document}

\begin{titlepage}

\maketitle

\input{abstract}

\thispagestyle{empty}
\end{titlepage}

\input{introduction}

\input{notation}

\input{diffusion}

\input{next_deriv}

\input{laplacian}

\input{subgradient}

{
\bibliography{dihyper}
\bibliographystyle{alpha}
}

\end{document}

%% file: abstract.tex
\begin{abstract}
In spectral graph theory, the Cheeger's inequality gives upper and lower bounds of edge expansion in normal graphs in terms of the second eigenvalue of the graph's Laplacian operator. Recently this inequality has been extended to undirected hypergraphs and directed normal graphs via a non-linear operator associated with a diffusion process in the underlying graph. 

In this work, we develop a unifying framework for defining a diffusion operator on a directed hypergraph with stationary vertices, which is general enough for the following two applications. 

1. Cheeger's inequality for directed hyperedge expansion. 

2. Quadratic optimization with stationary vertices in the context of semi-supervised learning.

Despite the crucial role of the diffusion process in spectral analysis, previous works have not formally established the existence of the corresponding diffusion processes.  In this work, we give a proof framework that can indeed show that such diffusion processes are well-defined.  In the first application, we use the spectral properties of the diffusion operator to achieve the Cheeger's inequality for directed hyperedge expansion. In the second application, the diffusion operator can be interpreted as giving a continuous analog to the subgradient method, which moves the feasible solution in discrete steps towards an optimal solution.
\end{abstract}

%% file: introduction.tex
\section{Introduction}
\label{sec:intro}

Classical spectral graph theory relates the edge expansion properties of a normal graph $G=(V,E)$ with the eigenvalues and eigenvectors of
matrices appropriately defined for the graph~\xcite{alon1985lambda1, alon1986eigenvalues, arora2010subexponential, louis2012many, lee2014multiway}.  
Intuitively,
for a given graph, edge expansion gives a lower bound
on the ratio of the weight of edges leaving a subset~$S$ of vertices to the weight of $S$.
Therefore, these notions have applications in graph partitioning or clustering~\xcite{jacm/KannanVV04,colt/MakarychevMV15,PengSZ15}, in which a graph is partitioned into clusters such that, loosely speaking,
the general goal is to minimize the number of edges crossing different clusters with respect to the cluster sizes.

The reader can refer to the standard references~\xcite{chung1997spectral, hoory2006expander} for a comprehensive description on spectral graph theory.
In short, the celebrated
Cheeger's inequality~\xcite{alon1985lambda1, alon1986eigenvalues} gives
upper and lower bounds of the edge expansion $\phi_G$ in terms
of the second eigenvalue $\gamma_2$ of the graph Laplacian as follows:
$$\frac{\gamma_2}{2} \leq \phi_G \leq \sqrt{2\gamma_2}.$$

Recently, higher-order Cheeger inequalities~\xcite{lee2014multiway,kwok2015improved} have been achieved
to relate higher order eigenvalues of the graph Laplacian
with multi-way edge expansion, and conditions in which spectral clustering
can be applied have been explored~\xcite{PengSZ15}.

For more general graph models,
Cheeger's inequality has been extended
to undirected hypergraphs~\xcite{louis2015hypergraph,chan2016spectral}
and directed normal graphs~\xcite{yoshida2016nonlinear}.
The high level approach is to define a diffusion operator
$\L_\omega: \R^V \ra \R^V$ and use properties of the
diffusion process $\frac{d f}{d t} = - \L_\omega f \in \R^V$
to achieve variants of the Cheeger's inequality.
%

In this paper, we develop a formal spectral framework to analyze
edge expansion and related problems in directed hypergraphs~\xcite{gallo1993directed}, which 
are general enough to subsume all previous graph models.

\noindent \textbf{Directed Hypergraphs.}
We consider an edge-weighted \emph{directed  hypergraph} $H=(V,E,w)$, where~$V$ is the vertex set
and~$E\subseteq 2^V \times 2^V$ is the set of directed hyperedges.
Each directed hyperedge $e\in E$ is denoted by $(\Te, \He)$, where $\Te\subseteq V$ is the \emph{tail} and $\He\subseteq V$ is the \emph{head};
we assume that both the tail and the head are non-empty, and
we follow the convention that the direction is from tail to head.
The function $w: E \ra \R_+$ assigns a non-negative weight to each edge.
Note that $\Te$ and $\He$ do not have to be disjoint.
This notion of directed hypergraph was first introduced
by Gallol et al.~\xcite{gallo1993directed},
who considered applications in propositional logic, analyzing dependency in relational database, and traffic analysis.

Observe that this model captures previous graph models:
(i) an undirected hyperedge $e$ is the special case when $\Te = \He$, and (ii) a directed normal edge $e$ is the special case when $|\Te| = |\He| = 1$.

\noindent \textbf{Directed Hyperedge Expansion.}
In addition to edge weights,
each vertex $u \in V$ has weight $\omega_u := \sum_{e \in E: u \in \Te \cup \He} w_e$ that is also known as its \emph{weighted degree}.
Given a subset $S \subseteq V$, denote $\overline{S} := V \setminus S$ and $\omega(S) := \sum_{u\in S}\omega_u$.
Define the out-going cut $\partial^+(S) := \{e \in E:  \Te \cap S \neq \emptyset \wedge \He \cap \overline{S} \neq \emptyset\}$,
and the in-coming cut $\partial^-(S) :=
\{e \in E:  \Te \cap \overline{S} \neq \emptyset \wedge \He \cap S \neq \emptyset\}$.
The out-going edge expansion of $S$ is $\phi^+(S) := \frac{w(\partial^+(S))}{\omega(S)}$,
and the in-coming edge expansion is $\phi^-(S) := \frac{w(\partial^-(S))}{\omega(S)}$.
The edge expansion of $S$ is $\phi(S) := \min \{\phi^+(S), \phi^-(S)\}$.
The edge expansion of $H$ is 
$$\phi_H := \min_{\emptyset \neq S \subset V:
\omega(S) \leq \frac{\omega(V)}{2}} \phi(S).$$
This notion of directed hyperedge expansion is consistent with
the edge expansion defined for undirected normal graphs~\xcite{chung1997spectral}, directed normal graphs~\xcite{yoshida2016nonlinear}
and undirected hypergraphs~\xcite{louis2015hypergraph,chan2016spectral}.

In this paper, we give a unifying framework of diffusion process
in directed hypergraphs, and relate the spectral properties
of the associated diffusion operator with edge expansion, thereby achieving a variant of the Cheeger's inequality
for directed hypergraphs.  Our diffusion framework is general enough to 
give a continuous analog of the subgradient method used in
in the context of semi-supervised learning.

\subsection{Our Contributions and Techniques}

\noindent \textbf{Spectral Analysis on Directed Hyperedge Expansion via Diffusion Process.}
We give a unifying framework of diffusion process in directed hypergraphs,
which is general enough for the following two applications.

\noindent \emph{Application 1: Cheeger's Inequality for Directed Hyperedge Expansion.}
We shall see that the spectral properties
of the diffusion operator can give a variant of the Cheeger's inequality
for directed hyperedge expansion.

\noindent \emph{Application 2: Quadratic Optimization with Stationary Vertices.}
An edge-weighted directed hypergraph $H$ induces
the following quadratic form
$$\mathsf{Q}(f) = \frac{1}{2} \sum_{e\in E} w_e \max_{(u,v)\in \Te\times \He}([f_u - f_v]^+)^2,$$
where for $x \in \R$, $[x]^+ := \max\{x,0\}$.  In semi-supervised learning
setting~\xcite{hein2013total}, vertices have some underlying label $f^* \in \R^V$.
There is a subset $T \subset V$ of \emph{stationary vertices} for which the labels $f^*_T \in \R^T$ are known.
Higher-order causal relationships between vertices can be modeled by directed hyperedges.
The task is to predict the labels for the unlabeled vertices $N := V \setminus T$
by finding $f \in \R^V$ such that
$\Q(f)$ is minimized, subject to $f_T = f^*_T$.

Since $f_T \in \R^T$ is fixed, we can view
$\Q : \R^N \ra \R$ as a function on $f_N \in \R^N$.
Since $\Q$ is convex but not differentiable everywhere,
the subgradient method~\xcite{Shor1985} can be used to solve the optimization problem.
However, in general, the subgradient does not necessarily lead to a descent direction
for $\Q$.  We shall show that the diffusion operator gives a subgradient of $\Q$ which also
leads to a descent direction.  Hence, the diffusion process is theoretically interesting
because it can be viewed as a continuous version of the subgradient method in which
$\Q$ is always decreasing towards the minimum value.

We give a unifying framework of diffusion process on directed hypergraphs
that generalizes those for undirected hypergraphs~\xcite{louis2015hypergraph,chan2016spectral}
and directed normal graphs~\xcite{yoshida2016nonlinear}.
Loosely speaking, the vertices are associated with a
\emph{density vector}~$f \in \R^V$.
The \emph{discrepancy} of an edge $e \in E$ with respect to $f$
is $\Delta_e(f) := \max_{u \in \tail_e} f_u - \min_{v \in \head_e} f_v$.
As in the case for directed normal graphs,
an edge~$e$ is \emph{active} with respect to~$f$ if~$\Delta_e(f) > 0$,
in which case edge~$e$ causes a rate of~$w_e \Delta_e(f)$ measure flowing
from vertices in~$\tail_e$ having maximum density to those
in~$\head_e$ having minimum density.

\noindent \textbf{New Feature: Stationary Vertices.}
Since in our second application the stationary vertices~$T$
have fixed~$f$ values, our diffusion process needs to incorporate
the property that the densities of stationary vertices do not change.
As in previous approaches,
each non-stationary vertex~$u \in N$ has some positive
weight~$\omega_u$.  In the first application
of hyperedge expansion, for $u \in V = N$, we set~$\omega_u := \sum_{e \in E: u \in \Te \cup \He} w_e$;
in the second application of quadratic optimization,
we shall see that setting~$\omega_u := 1$ for $u \in N$ will be appropriate.
Non-stationary vertices are also associated with 
a \emph{measure vector}~$\vp_N \in \R^N$
such that for each~$u \in N$, the invariant~$\vp_u =  \omega_u f_u$ holds;
hence, the density of a non-stationary vertex will change
according to its weight after releasing or absorbing measure (from other vertices).
On the other hand, a stationary vertex is not associated with any weight,
and the new feature is that a stationary vertex's density will not change even 
after releasing or absorbing measure.
The formal rules of our diffusion process are given in Definition~\ref{defn:rules}.

\noindent \textbf{Existing Approaches and Inadequacies.}
Given the current density vector $f \in \R^V$,
we only need to consider active edges as in~\xcite{yoshida2016nonlinear}.
Then, as in~\xcite{louis2015hypergraph,chan2016spectral},
the weight $w_e$ of an active hyperedge~$e$ needs
to be distributed carefully among the pairs $\Te \times \He$
to form a symmetric matrix $A_f \in \R^{V \times V}$,
where the diagonal entries are chosen such that
each row corresponding to $u \in N$ sums to the weight $\omega_u$.
Since the $f$ values of stationary vertices in~$T$ do not change,
the rows of $A_f$ corresponding to~$T$ are actually unimportant.
Then, using the projection operator $\Pi_N : \R^V \ra \R^N$ to keep
only coordinates corresponding to $N$, we can define the diffusion process by:
$$\frac{d f_N}{d t} = - (\Pi_N - \W^{-1} \Pi_N A_f) f \in \R^N,$$
where $\W \in \R_+^{N \times N}$ is the diagonal matrix
whose $(u,u)$-entry equals the weight $\omega_u$ of $u \in N$.
The diffusion operator $\L_\omega: \R^V \ra \R^N$
can also be defined as $\L_\omega f := - \frac{df}{dt}$.
With a bit of effort, we show that the techniques
in~\xcite{louis2015hypergraph,chan2016spectral} can be extended to handle
both directed hyperedges and stationary vertices.  However,
before we can use this diffusion process in our applications, 
below is an obvious question.

\noindent \textbf{Why would such a diffusion process exist?} It is
shown in~\xcite{chan2016spectral} that
the initial attempt~\xcite{louis2015hypergraph}
to distribute the weight $w_e$ of a hyperedge edge $e$ to form $A_f$ (which can change as
the density vector~$f$ changes)
might not lead to a well-defined diffusion process.
Moreover, even with the proper definition of the 
diffusion operator $\L_\omega: \R^V \ra \R^V$
in the absence of stationary vertices,
one could not simply pick an arbitrary linear subspace $S$ in $\R^V$
and hope that the diffusion process
$\frac{df}{dt} = - \Pi_S \L_\omega f$
is well-defined, where $\Pi_S : \R^V \ra \R^V$ is the
orthogonal projection operator
into the subspace $S$.  This is actually the subtle reason
why higher-order spectral properties do not hold for the diffusion operator
for hypergraphs.
On the other hand, the diffusion process with $\Pi_S$ is actually
well-defined in directed normal graphs~\xcite{yoshida2016nonlinear}.

In previous works, the existence of such diffusion processes have almost
never been properly discussed.  In~\xcite{yoshida2016nonlinear},
it was simply ``stipulated'' that the above diffusion process on directed normal graphs
exists. 
The approach for undirected hypergraphs in~\xcite{chan2016spectral} seems reasonable, and
some reader has also commented that it
is a ``natural first-guess fix'' to the initial attempt in~\xcite{louis2015hypergraph}.  However, in both works on undirected hypergraphs, properties of the diffusion process are used without first proving their existence.

\noindent \textbf{Significance of Our Results.}
Even though similar diffusion processes have been explored in previous works~\cite{louis2015hypergraph,yoshida2016nonlinear}, in most cases, the existence of the underlying process was not proved formally, which in some cases have led to inaccurate results or proofs~\cite{louis2015hypergraph}.  
Hence, it is evident that more theoretical justification is required for our current
framework of directed hypergraphs, especially with the more complicated rules
involving stationary vertices.
A theoretical contribution of this paper
is to give a formal treatment for the existence of such diffusion processes.
Therefore, even though the proofs in the submission might seem pedantic at first sight, they are valuable in showing the community how models should be correctly proved in order to avoid pitfalls in proofs.

\begin{theorem}[Existence of Diffusion Process (Section~\ref{sec:diffusion_exists})]
\label{th:intro_exist}
Given an edge-weighted directed hypergraph $H=(V,E,w)$
and vertex weights $\omega: N \ra \R_+$ for non-stationary vertices in $N \subseteq V$,
there exists a unique diffusion operator $\L_\omega: \R^V \ra \R^N$ satisfying
Definition~\ref{defn:rules} such that the diffusion 
process $\frac{d f_N}{d t} = - \L_\omega f \in \R^N$ is well-defined.
\end{theorem}

\noindent \emph{Terminology Remarks.}  For the case when $V = N$,
to be consistent with the literature,
the (weighted) Laplacian $\mathsf{M}: \R^V \ra \R^V$
is actually given by 
$\mathsf{M} := \W \L_\omega$.  Hence, to avoid confusion,
we will refer to the diffusion operator whenever possible.

\noindent \textbf{Proof Intuition.} The trickiest aspect of the framework is actually
because of the hyperedges.  Suppose at some instant, the current density vector is $f \in \R^V$.
The easy case is when all coordinates of $f$ are distinct, because in this case each (directed) hyperedge just behaves like a normal edge.  As the diffusion process
changes $f$ continuously, the resulting matrix $A_f$ in the above discussion actually remains unchanged in some short period of time, and this is enough to prove
the existence of the diffusion process.  As shown in~\xcite{chan2016spectral},
even when there are ties in the coordinates of $f$,
the first derivative $f^{(1)}$ (with respect to time)
can be uniquely determined, assuming implicitly that the diffusion process is well-defined.
In this paper, we go several (actually countably infinitely many) steps further,
and show that the diffusion rules uniquely determine all higher order derivatives~$f^{(i)}$.
Therefore, a lexicographical order on~$(f, f^{(1)}, f^{(2)}, \ldots)$
determines an equivalence relation~$\sigma$ on~$V$,
such that vertices in the same equivalence class will have the same $f$ values in
infinitesimal time.  Hence, each equivalence class behaves like a \emph{super vertex},
and for a short period of time, each hyperedge behaves like a normal edge between
super vertices.  As a result, the trajectory of the diffusion process can be solved
explicitly for a short period of time, until different equivalence classes merge or an equivalence class splits. Even though our high level proof strategy is intuitive,
as we shall see in Section~\ref{ssec:next-derivative},
a formal realization of our ideas requires some non-trivial mathematical constructs.

\vspace{5pt}

\noindent \textbf{(1) Application to Spectral Analysis of Directed
Hyperedge Expansion.}  After formally establishing the existence
of the diffusion process in Theorem~\ref{th:intro_exist},
we can relate directed hyperedge expansion with the spectral
properties of the diffusion operator, under the special case
with no stationary vertices.

\begin{theorem}[Cheeger's Inequality for Directed Hyperedge Expansion (Section~\ref{sec:spectral})]
\label{th:intro_cheeger}
Given an edge-weighted directed hypergraph $H=(V,E,w)$ and
vertex weights $\omega_u := \sum_{e \in E: u \in \Te \cup \He} w_e$,
the discrepancy ratio is
$\D(f) := \frac{2 \Q(f)}{\sum_{v \in V} \omega_v f^2_v}
= \frac{\sum_{e \in E: \Delta_e(f)>0} \Delta_e(f)^2}{\sum_{v \in V} \omega_v f^2_v}$.

Define $\gamma_2 := \min_{0 \neq f \in \R^V: \sum_{v \in V} \omega_v f_v = 0} \D(f)$.
Then, $\gamma_2$ is an eigenvalue of the diffusion operator
$\L_\omega: \R^V \ra \R^V$ from Theorem~\ref{th:intro_exist},
where any minimizer achieving $\gamma_2$ is a corresponding eigenvector.

Moreover, the eigenvalue $\gamma_2$ relates to the directed hyperedge expansion
$\phi_H$ in the following inequality:
$$\frac{\gamma_2}{2} \leq \phi_H \leq 2 \sqrt{\gamma_2}.$$

\end{theorem}

\noindent \textbf{Remarks.}  As in the cases
for undirected hypergraphs~\xcite{louis2015hypergraph,chan2016spectral}
and directed normal graphs~\xcite{yoshida2016nonlinear},
the upper bound in the Cheeger's inequality is slightly weaker
than that for undirected normal graphs.  We shall in fact concentrate 
on establishing the relationship between the discrepancy ratio $\D(\cdot)$
and the diffusion operator $\L_\omega$, after which the Cheeger's inequality
actually follows from nearly the same argument as previous works.

\vspace{5pt}

\noindent \textbf{(2) Continuous Analog of the Subgradient Method for
Quadratic Optimization with Stationary Vertices.}
The relationship between the diffusion operator $\L_\omega$ and the (sub)gradient
of the quadratic form $\Q$ has been mentioned as a potential
research direction~\xcite{chan2016spectral}.  We make this connection more precise and show that
the diffusion operator does not give only a subgradient, but also a descent direction
with respect to $\Q$.

\begin{theorem}[Diffusion Operator Gives a Subgradient and Descent Direction (Section~\ref{sec:subgradient})]
Suppose an edge-weighted directed hypergraph $H=(V,E,w)$ with
stationary vertices $T$ (where $f_T \in \R^T$ is fixed)
induces the quadratic form $\Q: \R^N \ra \R$, which
is a function on the restricted vector $f_N \in \R^N$ for
the non-stationary vertices $N$.  Then, for any $f \in \R^V$,
the diffusion operator $\L_\omega: \R^V \ra \R^N$ gives a
subgradient of $\Q$ at $f_N \in \R^N$.

Moreover, in the diffusion process defined by $\frac{df_N}{dt} = - \L_\omega f$,
the value of $\Q$ is non-increasing and converges to the minimum value.
\end{theorem}

\noindent \textbf{Proof Intuition.}  Observe that $\Q$ is differentiable
at $f_N \in \R^N$ when the coordinates have distinct values and different
from those of the stationary vertices.  Without loss of generality, we may assume
that the vector $f_T \in \R^T$ for stationary vertices has distinct coordinates.
Otherwise, we could merge two stationary vertices if they have the same $f$ value.
Given a permutation~$\sigma$ of~$V$,
we denote by $\Q_\sigma$ the restriction of $\Q$ to vectors whose
coordinates are distinct and consistent with~$\sigma$.
It follows that for vectors $f$ whose coordinates are consistent with $\sigma$,
the two gradients coincide, i.e. $\nabla \Q(f) = \nabla \Q_\sigma(f)$.

The harder case is when the coordinates of $f$ are not distinct.
Given any $f \in \R^V$ whose coordinates might not be distinct,
let $\Sigma$ be the set of permutations that can result from resolving ties
in the coordinates of $f$.  Then, it is not too difficult to see that
for any $\sigma \in \Sigma$, $\nabla \Q_\sigma(f)$ is a subgradient of $\Q$ at $f$.
In Section~\ref{sec:subgradient}, we shall see that the most technical part is to
prove that there is a distribution $\mathcal{D}$ on $\Sigma$
such that $\L_\omega f = \expct_{\sigma \leftarrow \mathcal{D}} [\nabla \Q_\sigma(f)]$,
which implies that $\L_\omega f$ is a subgradient.

\vspace{5pt}

\noindent \textbf{Impacts on Semi-Supervised Learning.}
The quadratic optimization problem with stationary vertices has
application in semi-supervised learning.
For the special case of undirected hypergraphs,
Hein et al.~\xcite{hein2013total} solved this problem by a more complicated
primal-dual approach.  Our result implies that this optimization problem 
has a deeper connection with the diffusion process, which gives a subgradient
that is also a descent direction with respect to the objective function.
In our companion experimental paper~\cite{ZhangHTC17},
we explore this connection to implement a simpler solver
for the quadratic optimization problem.


\subsection{Related Work}

\noindent \emph{Other Works on Spectral Graph Theory.}
Besides the most related works~\xcite{louis2015hypergraph,chan2016spectral,yoshida2016nonlinear}
that we have already mentioned,
there were also previous attempts that considered
spectral analysis for hypergraphs~\xcite{chung1993laplacian, friedman1995second, rodriguez2009laplacian}
and directed graphs~\xcite{chung2005laplacians, li2012digraph, bauer2012normalized}.
Spectral properties of $k$-uniform directed hypergraphs have also been studied via tensors in~\xcite{xie2016spectral}.

\noindent \emph{Higher-Order Cheeger Inequalities.}
For normal graphs,
Cheeger-like inequalities 
to relate higher-order spectral properties of graph Laplacians
with multi-way edge expansion have been investigated~\xcite{louis2011algorithmic, louis2012many, kwok2013improved, lee2014multiway, louis2014approximation, kwok2015improved}.
On the other hand, for hypergraphs,
it is shown~\xcite{chan2016spectral} that higher-order spectral properties of the diffusion operator
cannot be related to the discrepancy ratio as in the statement of Theorem~\ref{th:intro_cheeger}.  However, Cheeger-like inequalities
have been derived in terms of the discrepancy ratio~\xcite{chan2016spectral},
but not related to the spectral properties of the diffusion operator.

\noindent \emph{Other Treatments of Directed Graphs or Hypergraphs.}
A common way to handle directed normal graphs or undirected hypergraphs is to convert the underlying graph into
an undirected normal graph.  For instance, 
by considering the time-reversibility of the stationary distribution of a random walk process,
Chung~\xcite{chung2005laplacians} essentially converts a directed graph into an undirected graph.
One treatment of hypergraphs~\xcite{zhou2006learning} in the learning community is to replace each hyperedge by a clique or a star, thereby reducing a hypergraph to a normal graph.  However, such a treatment differentiates
how evenly a hyperedge is cut, whereas in the original definition of hyperedge expansion,
the penalty for cutting a hyperedge is the same, regardless of how evenly it is cut.

\noindent \emph{Quadratic Form.}
Hein et al.~\xcite{hein2013total} approached the semi-supervised learning problem
on hypergraphs via a quadratic optimization, whose objective is equivalent to the quadratic form considered in~\xcite{louis2015hypergraph,chan2016spectral}.
The quadratic form we propose for directed hypergraphs is a generalization of that for  undirected normal graphs~\xcite{zhu2003semi, zhou2004learning}, directed normal graphs~\xcite{yoshida2016nonlinear} and undirected hypergraphs~\xcite{hein2013total}.

\vspace{5pt}

\noindent \textbf{Paper Organization.}  We have mentioned our main results and their proof intuition
in the introduction.
We define a diffusion process for directed hypergraphs in Section~\ref{sec:prelim}, and outline its existence proof in
Section~\ref{sec:diffusion_exists},
while the complete technical proof of its existence is given in Section~\ref{ssec:next-derivative}. 
  The two applications of the diffusion processes are given in Sections~\ref{sec:spectral} and~\ref{sec:subgradient}.

%% file: notation.tex
\section{Diffusion Process on Directed Hypergraphs}
\label{sec:prelim}

We consider an edge-weighted \emph{directed  hypergraph} $H=(V,E,w)$, where $V$ is the vertex set and $E\subseteq 2^V \times 2^V$ is the set of directed hyperedges.
%
%
Each vertex is associated with some value in $\R$,
which is represented by a vector $f \in \R^V$.
For each vertex $u$, we use $f_u$ to denote the coordinate corresponding to $u$.
For a subset $S \subset V$, we use $f_S \in \R^S$ to indicate the restriction of $f$
to the coordinates corresponding to $S$.
Moreover, there are two types of vertices:
\begin{compactitem}
\item[1.] \textbf{Stationary Vertices (denoted as $T$).} The $f$ values of these vertices
do not change. The weight of a stationary vertex is not defined.  
We shall see later that it might be convenient 
to imagine that it has infinite weight.

\item[2.] \textbf{Non-Stationary Vertices (denoted as $N$).}
The $f$ values of these vertices can change.  Each non-stationary vertex $u$ is equipped
with a positive weight $\omega_u > 0$;
depending on the application,
the vertex weights may or may not depend on the edge weights.
We use $\W \in \R^{N \times N}$ to denote the diagonal matrix
whose $(u,u)$-th entry gives the weight~$\omega_u$ of vertex~$u$.
\end{compactitem} 


We use $\vec{1}$ to denote the all ones vector and $\I$ to denote the identity matrix,
with the dimensions determined from the context. 
For any $x \in \mathbb{R}$, we denote $\plus{x} := \max\{0,x\}$.

\noindent \textbf{Directed Hyperedge Expansion.}
Recall that directed hyperedge exapansion is defined in Section~\ref{sec:intro}.  We define
some useful parameters.

\begin{definition}[Edge Discrepancy]
\label{defn:edge_disc}
	Given $f \in \mathbb{R}^V$, for each $e\in E$, let $S_e(f) := \{ u\in \Te : f_u=\max_{v\in \Te}f_v \}$ be the set of vertices in $\Te$ with the maximum $f$ and $I_e(f) := \{ u\in \He : f_u=\min_{v\in \He}f_v \}$ be those in $\He$ with minimum $f$.
	The discrepancy of edge $e$ is
	\begin{equation*}
	\Delta_e(f) := \max_{u \in \tail_e} f_u - \min_{v \in \head_e} f_v.
	\end{equation*}	
	\end{definition}

\begin{definition}[Discrepancy Ratio]
\label{defn:disc_ratio}
Suppose there is no stationary vertex, i.e., $V = N$.
The discrepancy ratio of $f$ is defined as
	\begin{equation*}
	\mathsf{D}(f) := \frac{\sum_{e\in E} w_e ([\Delta_e(f)]^+)^2}{\sum_{u\in V} \omega_u f_u^2} = \frac{\sum_{e\in E} w_e \max_{(u,v)\in \Te\times \He}([f_u - f_v]^+)^2}{\sum_{u\in V} \omega_u f_u^2}.
	\end{equation*}
\end{definition}

\noindent \textbf{Relationship between Edge Expansion and Discrepancy Ratio.}
For $S \subseteq V$, let $\chi_S \in \{0,1\}^V$ denote its characteristic vector.
Then, $\phi^+(S) = \mathsf{D}(\chi_S)$ and $\phi^-(S) = \mathsf{D}(-\chi_S)$.
The numerator of the discrepancy ratio can also be studied in the following optimization problem.

%

\begin{definition}[Quadratic Optimization with Stationary Vertices]
Suppose $T\neq \emptyset$.
Given some fixed $f_T \in \R^T$, find $f_N \in \R^N$ such that
the following quadratic form is minimized:
	\begin{equation*}
	\mathsf{Q}(f) := \frac{1}{2}\sum_{e\in E} w_e ([\Delta_e(f)]^+)^2 = \frac{1}{2} \sum_{e\in E} w_e \max_{(u,v)\in \Te\times \He}([f_u - f_v]^+)^2.
	\end{equation*}
\end{definition}

We remind the reader that the diffusion process we consider has the following two applications.

\noindent \textbf{Application 1:} 
Define some operation whose spectral properties
are related to the directed hyperedge expansion $\phi_H$
of a directed hypergraph~$H$.



\noindent \textbf{Application 2:}  Since $\mathsf{Q}$ is a (non-differentiable) convex function with domain $\R^N$,
in the subgradient method, for any $f_N \in \R^N$, one wish to find
a subgradient of $\mathsf{Q}$ that is also a direction of descent.

\noindent \textbf{Diffusion Process.}  We consider 
a diffusion process (possibly with stationary vertices) that has implications for both Applications 1 and 2.
We consider the \textbf{projection operator} $\Pi_N : \R^V \ra \R^N$,
which can be represented by a matrix $\{0,1\}^{N \times V}$,
whose $(u,v)$-th entry is non-zero \emph{iff} $u=v \in N$.
%
Associated with the diffusion process are two spaces described as follows.


\begin{enumerate}
\item {\bf Density Space $\R^V$.}
The above vector $f \in \R^V$ resides in this space.
Intuitively, this vector gives the ``density'' of a vertex,
and we shall describe the precise rules such that 
``measure'' tends to flow from denser vertices to less dense vertices.
Observe the density of a non-stationary vertex can change, while
the density of a stationary vertex remains constant.
For vectors $f,g$ in $\mathbb{R}^V$ or $\R^N$, the inner product is defined as $\langle f, g \rangle_\omega := \sum_{u \in N} \omega_u f_u g_u$, i.e., the coordinates outside~$N$ are ignored.
The associated norm is $\| f \|_\omega := \sqrt{\langle f, f \rangle_\omega}$.
	We use $f \perp_\omega g$ to denote $\langle f, g \rangle_\omega = 0$.

	\item {\bf Measure Space $\R^N$.}
	Each non-stationary vertex has some net \emph{measure} (that could be negative). In the diffusion process, it is this measure that is being transferred between vertices.
	The measure vector $\vp_N \in \R^N$ is related to the density vector $f \in \R^V$
	by  $\vp_N := \W f_N$.  Hence, the density of a non-stationary vertex
	will increase if there is a net intake of measure.
	On the other hand, we can imagine that each stationary vertex can release or absorb unlimited amount of measure without changing its density.  In this work, we do not need to consider inner product in the measure space.
	
\end{enumerate}

We shall describe the diffusion rules such that the
derivative $\frac{d f}{d t}$ of the density vector with respect to
time $t$ is well-defined.
Then, we define a diffusion operator $\L_\omega : \R^V \ra \R^N$
by $\L_\omega f = - \Pi_N \frac{d f}{d t}$, and the operator
is related to our applications as follows.

\begin{enumerate}
\item[1.] \textbf{Spectral Analysis of Directed Hyperedge Expansion.}
For $V = N$, we shall see that the diffusion operator $\L_\omega$
can be related to the discrepancy ratio by
$\frac{\langle f , \L_\omega f \rangle_\omega}{\langle f , f \rangle_\omega}
= \D(f)$.  Moreover, the second eigenvalue of $\L_\omega$ is well-defined,
and relates to directed hyperedge expansion via a Cheeger-like inequality.

\item[2.] \textbf{Subgradient of Quadratic Form.} In the quadratic form
$\Q$, there is no obvious choice of weight for each (non-stationary) vertex.
If we just let each non-stationary vertex have unit weight,
then we shall show that $\L_\omega f$ gives a subgradient of $\Q$.  Moreover,
in this case $\L_\omega f = - \Pi_N \frac{df}{dt}$ gives a descent direction,
as we shall show that $\frac{d \Q}{d t} \leq 0$ in the diffusion process.
\end{enumerate}


\begin{definition}[Rules of Diffusion Process]
\label{defn:rules}
Suppose at some instant the system is in a state
given by the density vector $f \in \R^V$.  Then, at this instant,
measure is transferred between vertices according to the following rules.
For $u \in \tail_e \cup \head_e$,	let $\vp^{(1)}_u(e)$ be the net rate of measure flowing into vertex~$u$ due to edge~$e$.


\begin{compactitem}

\item [\textsf{R(0)}] For each non-stationary vertex $u \in N$,
		$\omega_u \frac{d f_u}{d t} = \vp^{(1)}_u := \sum_{e \in E: u \in \tail_e \cup \head_e} \vp^{(1)}_u(e)$, i.e., the density will change according to the net rate of incoming measure flow divided by its weight.
		
		For each stationary vertex $u \in T$, $\frac{d f_u}{d t} = 0$, i.e.,
		the density does not change.  Hence, a stationary vertex can absorb or release measure without changing its density.

\item [\textsf{R(1)}] 
		We have (i) $\vp^{(1)}_u(e) < 0$ only if $u \in S_e(f)$ and $\Delta_e(f) > 0$, and (ii) $\vp^{(1)}_u(e) > 0$ 
		only if $u \in I_e(f)$ and $\Delta_e(f) > 0$.

		\item [\textsf{R(2)}] Each edge $e \in E$ such that
		$\Delta_e(f) > 0$ causes a rate $w_e \cdot \Delta_e(f)$
		of measure flowing from vertices in $S_e(f)$ to vertices in $I_e(f)$.
%
		Specifically, $- \sum_{u \in S_e(f)} \vp^{(1)}_u(e)  = w_e \cdot \Delta_e(f) = \sum_{u \in I_e(f)} \vp^{(1)}_u(e)$.

%
%
		
		%
	\end{compactitem}
\end{definition}

\noindent \textbf{Explanation.} We remark that in our system, every vertex
has a density, but only a non-stationary vertex contains measure, which
is related to the density with respect to its weight.  Even though
a stationary vertex $u$ absorbs or releases measure without changing its density,
we still consider its net intake rate $\vp^{(1)}_u(e)$ of measure 
due to some edge $e$ to facilitate the description of rules
\textsf{(R1)} and~\textsf{(R2)}.

\noindent \textbf{Relationship with Undirected Graph.}
As in previous approaches~\xcite{chan2016spectral}, we shall later show that
our diffusion rules uniquely define $\frac{df}{dt}$, even though
the $\vp^{(1)}_u(e)$'s might not be unique.  The diffusion process can be interpreted
via an undirected graph with dynamic weights as follows.

Suppose for some edge $e \in E$ such that $\Delta_e(f) > 0$, the rates $\vp^{(1)}_u(e)$ for $u \in \tail_e \cup \head_e$
are determined.  Then, by rule~\textsf{(R2)} one can
consider a complete bipartite on $(S_e, I_e)$ to assign some $a^e_{uv}$ for $(u,v) \in S_e \times I_e$
such that $\sum_{(u,v) \in S_e \times I_e} a^e_{uv} = w_e$.
Moreover, for $v \in I_e$, 
$\vp^{(1)}_v(e) = \sum_{u \in S_e} a^e_{uv} \Delta_e(f)$;
for $u \in S_e$, 
$\vp^{(1)}_u(e) = - \sum_{v \in I_e} a^e_{uv} \Delta_e(f)$.

The interpretation is that there is a measure flowing from $u \in S_e$ to $v \in I_e$
due to edge $e$ with rate $a^e_{uv} \cdot \Delta_e(f)$.
For notational convenience, we write $a_{uv}^e = a_{vu}^e$, 
and assume that for all $u, v \in V$ and $e \in E$, $a^e_{uv}$ is well-defined and 
can be set to zero if necessary.

Then, we can construct a symmetric matrix $A_f \in \R^{V \times V}$ such that
for $u \neq v$, $A_f(u,v) := \sum_{e \in E} a^e_{uv}$.
Moreover, for $u \in N$, the diagonal entry $A_f(u,u)$ is chosen such that the entries of row~$u$ sum up to $\omega_u$; the rows corresponding to vertices in $T$ are not important, because we shall apply the projection
operator $\Pi_N$.
Even though the matrix $A_f$ might not be unique,
we shall see later that as long as Definition~\ref{defn:rules}
is satisfied, Theorem~\ref{th:exists} states
that the derivative $\frac{d f}{d t}$ is uniquely determined.
Hence, we can define
the diffusion operator $\L_\omega : \R^V \ra \R^N$ as follows:
$$ \L_\omega f := - \Pi_N \frac{df}{dt},$$ 
which equals $(\Pi_N - \W^{-1} \Pi_N A_f) f$.

%% file: diffusion.tex
\section{Existence of Diffusion Process}
\label{sec:diffusion_exists}

The main result in this section is
to establish the existence of a diffusion process
satisfying the rules in Definition~\ref{defn:rules}.

\begin{theorem}[Existence of Diffusion Process]
\label{th:exists}
There exists a diffusion process satisfying
the rules given in Definition~\ref{defn:rules}.
Moreover, at any time $t$ when the system
is in the state given by the density vector $f \in \R^V$,
the rules uniquely determine $\frac{df}{dt}$.
\end{theorem}

\noindent \textbf{Proof Strategy.} In Definition~\ref{defn:rules},
rule~\textsf{R(1)} means that for an edge $e \in E$,
there should be some measure flowing from vertices with highest density
in $\tail_e$ to those with lowest density in $\head_e$, provided
that the former vertices have higher density than the latter ones according to rule~\textsf{R(2)}.  For normal (directed) graphs, this is usually straightforward, and not much emphasis is placed on whether
the diffusion process is well-defined~\xcite{yoshida2016nonlinear}.  However, as shown in~\xcite{chan2016spectral}, the situation is more subtle when hyperedges are involved.  

The reason is that for a hyperedge $e$, there can
be multiple number of vertices in $S_e(f)$ that give measure due to edge $e$.
The tricky part is that all the vertices in $S_e(f)$ that contribute to releasing flow due to $e$ must have the same density in infinitesimal time.  Even though
a subset of vertices have the maximum density in $S_e(f)$ at some instant,
if there is one vertex whose density is just about to grow higher than other vertices in infinitesimal time, then this vertex should actually be the only vertex that contributes to giving measure due to $e$.  The situation gets even more complicated because a vertex can participate in various hyperedges, playing both the giving and receiving roles.

This complication arises because we do not know in infinitesimal time, which vertices in $S_e(f)$ will actually have the highest density in infinitesimal time (and similarly which vertices in $I_e(f)$ will have the lowest density).
In~\xcite{chan2016spectral}, a subroutine for a version of the densest subset problem
is used to determine the first derivative of $f$, assuming that the diffusion process exists.  While this is sufficient to give a definition of $\L_\omega$,
the existence of the diffusion process is not formally proved.  

We observe that
all higher derivatives can also be uniquely determined.
Hence, a lexicographical order on the derivatives can be used to
partition vertices with an equivalence relation
such that vertices in the same equivalence class
will have the same density in infinitesimal time.
Intuitively, vertices in the same equivalence class behave like a super vertex,
and in infinitesimal time, each (directed) hyperedge behaves just like a normal edge.  Therefore, the trajectory of $f$ in infinitesimal time can be solved explicitly (in Lemma~\ref{lemma:distinct_densities}) and this establishes the existence of the diffusion process.
Strictly speaking, this only proves that starting from any density vector $f \in \R^V$, the diffusion process is well-defined in infinitesimal time, i.e.,
there exists some $\epsilon > 0$ (depending on $f$) for which the diffusion process can run for time $\epsilon$.  We remark that this is enough for the purpose of tackling our problems. However, once it is established that the diffusion process is well-defined in infinitesimal time, it can be shown that the whole process is also well-defined.

\begin{lemma}[Diffusion Process with Distinct Densities]
\label{lemma:distinct_densities}
Starting from some density vector $f \in \R^V$ with distinct
coordinates,
the trajectory of the diffusion process in infinitesimal time satisfying Definition~\ref{defn:rules} can be uniquely determined.
\end{lemma}

\begin{proof}
Since the diffusion process changes $f$ continuously,
in infinitesimal time,
the density vector $f \in \R^V$ still has distinct coordinates,
and the relative order of the $f$ values does not change. 
Throughout this proof, we consider the diffusion process only in this small time frame.

Hence, it follows that for each $e \in E$, $S_e(f)$ and $I_e(f)$
are both singletons and do not change; moreover,
the sign of $\Delta_e(f)$ also does not change.
Therefore, there exists a symmetric matrix $A \in \R^{V \times V}$
such that the diffusion process is described by the following differential equation:
$$ \frac{d f_N}{d t} = - (\Pi_N - \W^{-1} \Pi_N A) f.$$

Using the sub-matrices $A_N \in \R^{N \times N}$ and
$B \in \R^{N \times T}$, we can rewrite the differential equation as:
$$ \frac{d f_N}{dt} = - f_N + \W^{-1} A_N f_N + \W^{-1} B f_T.$$

Finally, we consider the transformation $f_N = \W^{-\frac{1}{2}} x_N$ to obtain:
$$ \frac{d x_N}{d t} = - \Lc x_N + \W^{-\frac{1}{2}} B f_T,$$
where $\W^{-\frac{1}{2}} B f_T \in \R^N$ is a constant vector,
and $\Lc := \I_N - \W^{-\frac{1}{2}} A_N \W^{-\frac{1}{2}}$
is a symmetric matrix, and therefore its eigenvectors form an orthonormal basis
for $\R^N$.  Hence, by considering each eigenspace independently, we can show that there is a unique solution for $x_N$ in infinitesimal time,
which implies that the trajectory for $f_N$ in infinitesimal time is also unique.
\end{proof}

\subsection{Ordered Equivalence Relation}

In order to have the right terminology to describe
which vertices will have their $f$ values stay together in
infinitesimal time, we use ordered equivalence relation.

\begin{definition}[Ordered Equivalence Relation]
\label{defn:ordered_equiv}
An ordered equivalence relation $\sigma$ on $V$
is an equivalence relation on $V$ such that
its equivalence classes are equipped with a total order.
We use $\mathcal{O}(V)$ to denote the collection of ordered
equivalence relations on $V$.
We use $\mathcal{C}[\sigma]$ to denote the collection 
of the equivalence classes of $\sigma$,
and $[u]_\sigma \in \mathcal{C}[\sigma]$ to denote the equivalence class containing $u \in V$.
The comparison operators $\prec, \preceq, =, \succeq, \succ$ are used in their usual sense
to indicate the order between two equivalence classes.
%
%
%
%
\end{definition}

\begin{definition}[Refinement]
\label{defn:refinement}
An ordered equivalence relation $\widehat{\sigma}$
is a refinement of $\sigma$ (denoted as~$\sigma \sqsubseteq \widehat{\sigma}$)
if the following holds.
\begin{compactitem}
\item[1.] For every $\widehat{S} \in \mathcal{C}[\widehat{\sigma}]$,
there exists $S \in \mathcal{C}[\sigma]$ such that $\widehat{S} \subseteq S$.

\item[2.] Suppose $\widehat{S}_1, \widehat{S}_2 \in \mathcal{C}[\widehat{\sigma}]$
and $S_1, S_2 \in \mathcal{C}[\sigma]$ are the corresponding equivalence classes
containing
them such that $S_1 \prec S_2$.  Then, we have $\widehat{S}_1 \prec \widehat{S}_2$.
\end{compactitem}
We say that $\widehat{\sigma}$ is at least as refined as $\sigma$.
\end{definition}

\begin{definition}[Compatibility, Consistency and Least Refinement]
\label{defn:compatibility}
An ordered equivalence relation $\sigma \in \mathcal{O}(V)$ is \textbf{compatible}
with a vector $f \in \R^V$ if every vertex in the same equivalence
class of $\sigma$ has the same $f$ value.
If in addition $[u]_\sigma \prec [v]_\sigma$ implies that $f_u \leq f_v$,
then we say that $\sigma$ is \textbf{consistent} with $f$.

We use $\sigma(f)$ to denote the least refinement of $\sigma$ that is compatible with $f$.
Observe that an equivalence class in $\sigma$ may be partitioned
into several equivalence classes in $\sigma(f)$ because the vertices have different $f$ values; the
relative order of these parts are given by their $f$ values.
\end{definition}

\noindent \textbf{Equivalence Class in Infinitesimal Time.}
Suppose we use $\varsigma$ to denote the trivial equivalence relation
in which $V$ is the only equivalence class.  Then,
the current density vector $f \in \R^V$ induces
an equivalence relation $\sigma_0 := \varsigma(f)$.
However, vertices in the same equivalence class of $\sigma_0$ could
have different $\frac{d f}{d t}$,
and hence will be in different equivalence classes in infinitesimal time.
This observation was made in~\xcite{chan2016spectral} to show that there is a unique
$f^{(1)}:=\frac{df}{dt}$ that can satisfy Definition~\ref{defn:rules}, if the
diffusion process exists.  
Recall that our goal is to prove that such a diffusion process indeed exists.
In order to use Lemma~\ref{lemma:distinct_densities} for this purpose,
we need to show that there is some ordered equivalence relation
that completely captures which vertices' densities will stay together in 
infinitesimal time.

Observe that the equivalence relation $\sigma_1 := \sigma_0(f^{(1)})$
might not be sufficient,
because vertices in the same equivalence class in $\mathcal{C}[\sigma_1]$
might be about to split because they have different second derivatives.
Our approach is to show that Definition~\ref{defn:rules}
actually uniquely determines all higher-order derivatives $f^{(i)} := \frac{d^i f}{d x^i}$.

This gives a chain of ordered equivalence relations defined by
$\sigma_i \sqsubseteq \sigma_{i+1} := \sigma_i(f^{(i+1)})$,
and the desired equivalence relation $\sigma^*$ is the least upper bound of this chain.
Alternatively, it can be viewed as an order on $V$ given by the lexicographical
order induced by the vectors $\{f^{(i)} : i \geq 0\}$, where 
we define $f^{(0)} := f$.

\subsection{Determining Higher Order Derivatives}

Starting from some density vector $f \in \R^V$,
we show that the diffusion process in Definition~\ref{defn:rules}
uniquely determines $f^{(i)} := \frac{d^i f}{d t^i} \in \R^V$
for all $i \geq 1$.
Recall that the vertices $T = V \setminus N$ are stationary and have zero higher derivatives.

\noindent \textbf{Notation.}  Given an ordered equivalence relation
$\sigma$ and $S \subseteq V$, we use $\max_{u \in S} [u]_\sigma$
to denote the ``largest'' equivalence class of $\sigma$ that contains
a vertex in $S$.  If $\sigma$ is compatible
with a vector $g \in \R^V$,  then the vertices in this equivalence
class have the same $g$ value, and 
we use $g(\max_{u \in S} [u]_\sigma)$ to denote this common $g$ value.
We can also define
$g(\min_{u \in S} [u]_\sigma)$ similarly.

\noindent \textbf{Active and Ambiguous Edges.}
For each edge $e \in E$, define 
$\Delta^{(0)}_e := \Delta_e(f) = f(\max_{u \in \tail_e} [u]_{\sigma_0}) - f(\min_{v \in \head_e} [v]_{\sigma_0})$.
The set of \emph{active} edges at level $0$ is $E^{(0)}_+ := \{e \in E: \Delta^{(0)}_e > 0\}$;
observe that these edges will remain active when we determine derivatives at higher levels.
The set of \emph{inactive} edges at level $0$ is 
$E^{(0)}_- := \{e \in E: \Delta^{(0)}_e < 0\}$; these edges will remain inactive
at higher levels.
The set of \emph{ambiguous} edges at level $0$ is
$E^{(0)}_0 := \{e \in E: \Delta^{(0)}_e = 0\}$;
at higher levels, each such edge could become active or inactive in some level $i$,
and once this happens, that edge will stay active or inactive for all higher levels.

\noindent \textbf{Higher Order Diffusion Rules.}  Definition~\ref{defn:rules}
gives the rules regarding the first derivative of $f$ with respect to time.
By applying the Envelope Theorem, we can derive higher order diffusion rules
as follows.

\begin{proposition}[Higher Order Diffusion Rules]
\label{prop:higher_rules}
The diffusion rules in Definition~\ref{defn:rules}
imply that for each $i \geq 0$,
the following rules hold.
For $u \in \tail_e \cup \head_e$,
we use
$\vp^{(i+1)}_u(e)$ to denote the $(i+1)$-st derivative
of the measure received by vertex~$u$ due to edge~$e$.

\begin{compactitem}

\item [\textsf{R(0)}] For each non-stationary vertex $u \in N$,
		$\omega_u f^{(i+1)}_u = \vp^{(i+1)}_u := \sum_{e \in E: u \in \tail_e \cup \head_e} \vp^{(i+1)}_u(e)$.
		
		For each stationary vertex $u \in T$, $f^{(i+1)}_u = 0$.

\item [\textsf{R(1)}] We have
$\vp^{(i+1)}_u(e) \neq 0$ only if
$e \in E^{(i)}_+$ and at least one of the following conditions hold:

(i) $u \in \arg \max_{v \in S^{(i)}_e} f^{(i+1)}_v$, or
(ii) $u \in \arg \min_{v \in I^{(i)}_e} f^{(i+1)}_v$.

		\item [\textsf{R(2)}] Each edge $e \in E^{(i)}_+$
		causes a contribution of $w_e \Delta^{(i)}_e$
		to the $(i+1)$-st derivative of the measure received by vertices
		in $I^{(i)}_e$, and a contribution of $-w_e \Delta^{(i)}_e$
		to that received by vertices in $S^{(i)}_e$.
		Specifically, $- \sum_{u \in S^{(i)}_e} \vp^{(i+1)}_u(e)  = w_e \cdot \Delta^{(i)}_e = \sum_{u \in I^{(i)}_e} \vp^{(i+1)}_u(e)$.
		
	\end{compactitem}
\end{proposition}

\begin{proofof}{Theorem~\ref{th:exists}}
Suppose at some instant in time,
the densities of the vertices are given by $f \in \R^V$.
Then, we shall prove later in Lemma~\ref{lemma:deriv}  that all higher derivatives
of $f$ at this moment can be uniquely determined
by Definition~\ref{defn:rules} and its extension Proposition~\ref{prop:higher_rules}.
This induces an equivalence relation $\sigma^*$ that captures which
vertices will stay together in infinitesimal time.

Therefore, each equivalence class in $\mathcal{C}[\sigma^*]$ acts
like a super vertex, where a super vertex is stationary if it contains a stationary vertex.
Since vertices from different equivalence classes will have different $f$ values in
infinitesimal time, we can apply Lemma~\ref{lemma:distinct_densities}
to show that the trajectory of the $f$ values can be uniquely determined in 
infinitesimal time.  This completes the existence proof for Theorem~\ref{th:exists}.
\end{proofof}

%% file: next_deriv.tex
\section{Generating the Next Derivative via a Densest Subset Problem}
\label{ssec:next-derivative}

In this section, we decribe the procedure $\mathfrak{D}^{(i+1)}$
to generate the $(i+1)$-st derivative.

\noindent \textbf{Defining $f^{(i+1)}$ Inductively.}
Suppose $\mathcal{F}_i := \{(f^{(j)}, \sigma_j, E^{(j)}_+, E^{(j)}_0)\}_{j=0}^i$
has already been constructed.
We shall show that
Proposition~\ref{prop:higher_rules}
implies that there is a procedure $\mathfrak{D}^{(i+1)}$
that returns a unique $f^{(i+1)} := \mathfrak{D}^{(i+1)}(\mathcal{F}_i)$
that is supposed to be the $(i+1)$-st derivative $\frac{d^{i+1} f}{d t^{i+1}}$.
Then, all other objects at level $i+1$ can be defined as follows.

\begin{compactitem}
\item $\sigma_{i+1} := \sigma_i(f^{(i+1)})$.
\item For $e \in E$, $\Delta^{(i+1)}_e := f^{(i+1)}(\max_{u \in \tail_e} [u]_{\sigma_{i+1}})
- f^{(i+1)}(\min_{v \in \head_e} [v]_{\sigma_{i+1}})$.

\item \textbf{Resolving Some Ambiguous Edges.}

$E^{(i+1)}_+ := E^{(i)}_+ \cup \{e \in E^{(i)}_0: \Delta^{(i+1)}_e > 0 \}$;

$E^{(i+1)}_0 := \{e \in E^{(i)}_0: \Delta^{(i+1)}_e = 0 \}$.

\end{compactitem}

We extend the framework in~\xcite{chan2016spectral}
to describe the procedure $\mathfrak{D}^{(i+1)}$ 
that takes the objects $\mathcal{F}_i$ from levels at most $i$
and returns the $(i+1)$-st derivative.  There are two major differences.
\begin{compactitem}
\item Since stationary vertices have constant $f$ values,
in some sense their weights can be perceived as infinity.  This
introduces some complication in the densest subset problem.
\item When we compute the $(i+1)$-st derivative, we need
to consider all objects constructed at levels at most $i$
in order to satisfy Definition~\ref{defn:rules}.
\end{compactitem}

\noindent \textbf{Consider Each Equivalence Class $U \in \mathcal{C}[\sigma_i]$ 
Independently.}  Vertices are in the same equivalence class of $\sigma_i$ because
derivatives of order up to $i$ cannot be used to separate them.  However,
it is possible that these vertices could be separated in infinitesimal time
due to the $(i+1)$-st derivative.  We consider the following densest subset problem.

\noindent \textbf{Densest Subset Problem.}
For each edge $e \in E$, 
define $c^{(i)}_e := w_e \Delta^{(i)}_e$ (which could be positive or negative), and
denote $S^{(i)}_e := \arg \max_{u \in \tail_e} [u]_{\sigma_i} \subseteq \tail_e$
and $I^{(i)}_e := \arg \min_{v \in \head_e} [v]_{\sigma_i} \subseteq \head_e$.
For any subset $X \subseteq U$,
we denote
$I^{(i)}_X := \{e \in E^{(i)}_+ : I^{(i)}_e \subseteq X\}$
and $S^{(i)}_X := \{e \in E^{(i)}_+ : S^{(i)}_e \cap X \neq \emptyset \}$.

\noindent \emph{Density at Level $i$.}
We consider a densest subset problem
in $(U, I^{(i)}_U, S^{(i)}_U)$,
and define a density function $\delta^{(i)}$
for each non-empty $X \subseteq U$ as follows.

$$
\delta^{(i)}(X) :=
\begin{cases}
\frac{\sum_{e\in I^{(i)}_X} c^{(i)}_e - \sum_{e\in S^{(i)}_X} c^{(i)}_e}{\sum_{v \in X} \omega_v}, \qquad & X \cap T = \emptyset,\\
0, \qquad &X\cap T\neq \emptyset.
\end{cases}
$$

\noindent \emph{Lexicographically Densest Subset.}
In order to see whether the equivalence class $U$ will be separated,
the procedure finds a maximal lexicographically densest subset $P$
with respect to the vector
$\vec{\delta}^{(i)}(P) := (\delta^{(0)}(P), \delta^{(1)}(P),
\ldots, \delta^{(i)}(P))$, where smaller indices have higher priorities.
This aims to find the vertices $P$ in $U$ with the maximum $(i+1)$-st 
derivatives, which will separate them from other vertices in $U$.

Below is an example (with $i=1$) illustrating why we consider the lexicographical density, instead of the density $\delta^{(i)}$ at the highest level $i$.
\begin{figure}[h]
	\centering
	\includegraphics*[width = 0.7\textwidth]{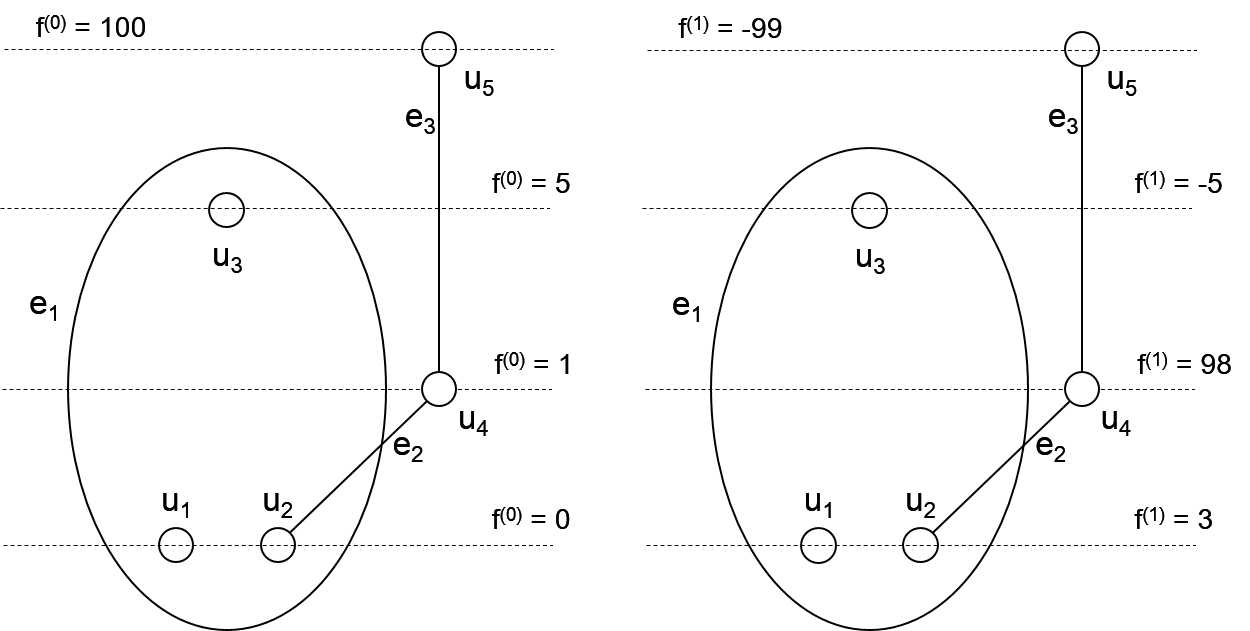}
	\caption{Given $f^{(0)}$ and $f^{(1)} = \frac{df^{(0)}}{dt}$, let $P = \{u_2\}$ and $P' = \{ u_1,u_2 \}$, then we have $\delta^{(0)}(P) = 1$, $\delta^{(1)}(P) = 95$, $\delta^{(0)}(P') = 3$ and $\delta^{(1)}(P') = 43.5$.
	In other words, compared with $P'$, $P$ has a larger level $1$ density but a smaller lexicographical density.
	Hence if we set $f^{(2)}_{u_2} = \delta^{(1)}(P) = 95$ (by considering only level $1$ density), then we would have a contradiction such that $f^{(0)}_{u_2} > f^{(0)}_{u_1}$ and $f^{(0)}_{u_2} < f^{(0)}_{u_1}$ both hold after infinitesimal time.}
	\label{fig:second_derivative_case}
\end{figure}

Observe for $0 \leq j \leq i$,
all vertices in $P$ are in the same equivalence class in $\sigma_j$.
Hence, $\delta^{(j)}(P)$ has already been defined in previous steps.
According to whether there are stationary vertices in $U$,
we will proceed differently.

\noindent \textbf{Case 1: No Stationary Vertices in $U$.}
The treatment is similar to~\xcite{chan2016spectral}.
After the maximal $P$ is identified above,
for every $u \in P$, we set $f^{(i+1)}(u) := \delta^{(i)}(P)$.
If $P \neq U$, then we recursively consider
the smaller instance $(\widehat{U}, I^{(i)}_{\widehat{U}}, S^{(i)}_{\widehat{U}})$,
where $\widehat{U} := U \setminus P$,
$I^{(i)}_{\widehat{U}} := I^{(i)}_{U} \setminus I^{(i)}_{P}$,
and $S^{(i)}_{\widehat{U}} := S^{(i)}_{U} \setminus S^{(i)}_{P}$.
Observe that in this smaller instance,
the density function $\widehat{\delta}^{(i)}$ will also be different.
It is proved in~Lemma 4.6 in \xcite{chan2016spectral} that
for any $\emptyset \neq X \subseteq \widehat{U}$, $\widehat{\delta}^{(i)}(X) \leq 
\delta^{(i)}(P)$.  Hence, in $\sigma_{i+1}$, vertices in $P$ will be separated away from
vertices in $\widehat{U}$.
Recall that the density function at level $i$
will be used for higher levels.  Observe that we should
use $\delta^{(i)}$ for subsets in $P$, and
if $\widehat{P}$ is the maximal subset found in this instance using
density $\widehat{\delta}^{(i)}$, then this new density function
will be used for vertices in $\widehat{P}$ in higher levels.

\noindent \textbf{Case 2: There is a Stationary Vertex in $U$.}
Since the equivalence class $U \in \mathcal{C}[\sigma_i]$ contains a stationary vertex,
this means that all vertices in $U$ have zero derivatives up to level $i$.
If the maximal $P$ found in the above has strictly positive
density $\delta^{(i)}(P) > 0$,
this means that $P$ does not contain any stationary vertex, and 
it will be separated from other vertices in $U$;
we also recursively consider the smaller instance as Case~1.

Observe that using $\delta^{(i)}$, the densest subset cannot have
negative density because by including a stationary vertex, the density becomes zero (refer to Figure~\ref{fig:negative_r_case}).
This means that we cannot use $\delta^{(i)}$ to detect vertices that 
are separated from the stationary vertices because their $(i+1)$-st derivatives should be negative.  Hence, if the maximal $P$ found has density $\delta^{(i)}(P) = 0$, we
consider the \emph{least densest subset problem} that aims to
find the vertices with the minimum $(i+1)$-st derivative.

\begin{figure}[H]
	\centering
	\includegraphics*[width = 0.3\textwidth]{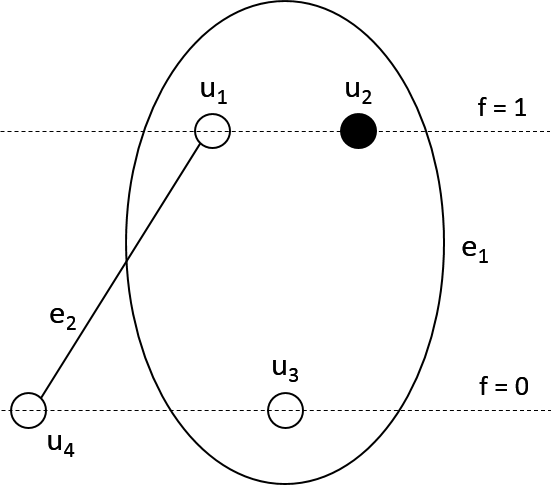}
	\caption{Suppose $T = \{u_2\}$ and $f^{(0)}_{u_1} = f^{(0)}_{u_2}=1$. By Definition~\ref{defn:rules} we have $f^{(1)}(u_1)<0$, but by including the stationary vertex $u_2$, we obtain a densest subset $\{u_1,u_2\}$ with zero density.}
	\label{fig:negative_r_case}
\end{figure}

\noindent \emph{Least Densest Subset Problem.}
We have a similar instance $(U, \widetilde{I}^{(i)}_U, \widetilde{S}^{(i)}_U)$.
For $X \subseteq U$, observe that we have different definitions
$\widetilde{I}^{(i)}_X := \{e \in E^{(i)}_+ : I^{(i)}_e \cap X \neq \emptyset \}$
and $\widetilde{S}^{(i)}_X := \{e \in E^{(i)}_+ : S^{(i)}_e \subseteq X\}$.

For each non-empty $X \subseteq U$,
the new density function $\widetilde{\delta}^{(i)}$ is defined as follows.

$$
\widetilde{\delta}^{(i)}(X) :=
\begin{cases}
\frac{\sum_{e\in \widetilde{I}^{(i)}_X} c^{(i)}_e - \sum_{e\in \widetilde{S}^{(i)}_X} c^{(i)}_e}{\sum_{v \in X} \omega_v}, \qquad & X \cap T = \emptyset,\\
0, \qquad &X\cap T\neq \emptyset.
\end{cases}
$$

Similar to before, the goal is then to find a maximal
lexicographically least densest subset $P$
with respect to the vector $(\widetilde{\delta}^{(0)}(P), \widetilde{\delta}^{(1)}(P),
\ldots, \widetilde{\delta}^{(i)}(P))$.

If $\widetilde{\delta}^{(i)}(P) < 0$, then we can
set for each $u \in P$, $f^{(i+1)}(u) = \widetilde{\delta}^{(i)}(P)$.
Then, we separate $P$ from $U$ and consider a smaller instance as before.
If $\widetilde{\delta}^{(i)}(P) = 0$, this means
all vertices in $U$ will stay together with the stationary vertex
in $\mathcal{C}[\sigma_{i+1}]$.

This completes the description of the procedure $\mathfrak{D}^{(i+1)}$.

\begin{lemma}[Diffusion Rules Determine Derivatives]
\label{lemma:deriv}
The diffusion rules in Proposition~\ref{prop:higher_rules}
uniquely determine all higher derivatives of $f \in \R^V$ with respect to time.
Moreover, for all $i \geq 0$, we have the following:

\begin{compactitem}
\item[1.] The
procedure $\mathfrak{D}^{(i+1)}$ gives the $(i+1)$-st derivative $f^{(i+1)}$.
\item[2.] We have $\norm{f^{(i+1)}}^2_\omega = - \sum_{e \in E^{(i)}_+} w_e \Delta^{(i)}_e \Delta^{(i+1)}_e$.
\end{compactitem}
\end{lemma}

\begin{proof}
We use a proof structure similar to~Lemma~4.8 in \xcite{chan2016spectral}.
Recall that we are considering a fixed instant in time.

\noindent \textbf{Separating an Equivalence Class $U \in \mathcal{C}[\sigma_i]$.}
The procedure $\mathfrak{D}^{(i+1)}$ takes all objects
$\mathcal{F}_i$ of levels up to $i$,
and attempts to see if the vertices in an equivalence class $U$ 
could be separated because their $(i+1)$-st derivatives are different.
We consider the subset $P \subseteq U$ with the maximum $(i+1)$-st derivatives,
and describe how the procedure $\mathfrak{D}^{(i+1)}$ could separate $U$
in the following cases.

\noindent \emph{Case 1: $P$ does not contain a stationary vertex.}
Suppose the vertices in $P$ have $\delta_M$ as their common $(i+1)$-st derivative.
By~\textsf{R(0)} and~\textsf{R(1)} in Proposition~\ref{prop:higher_rules},
we have

$$\omega(P) \cdot \delta_M = \sum_{u \in P} 
\sum_{e \in E^{(i)}_+: u \in I^{(i+1)}_e \cup S^{(i+1)}_e } \vp^{(i+1)}_u(e).$$

Since the vertices have the highest $(i+1)$-st derivatives,
by~\textsf{R(1)}, the vertices in $P$ must be responsible
for the edges in $I^{(i)}_P$ and $S^{(i)}_P$.

By~\textsf{R(2)},
it follows that the RHS of the above equation is:
$\omega(P) \cdot \delta_M = c^{(i)}(I^{(i)}_P) - c^{(i)}(S^{(i)}_P)$,
which implies that $\delta_M = \delta^{(i)}(P)$.

The next observation is that since $P$ will be separated
in infinitesimal time from other vertices in $U$,
in order to maintain consistency with previous levels up to $i$,
it follows that $P$ must be the maximal
lexicographically densest subset with respect to $\vec{\delta}^{(i)}$.

Hence, the procedure $\mathfrak{D}^{(i+1)}$ can successfully recover $P$ and
identify its $(i+1)$-st derivative.

\noindent \emph{Case 2: $P$ contains a stationary vertex.}
In this case, observe that the procedure will see that
the maximal lexicographically densest subset
with respect to $\vec{\delta}^{(i)}$ is actually the whole $U$.
However, in this case, the procedure
will consider the least densest subset problem.
Using a similar argument, it follows that the procedure
can recover the subset $\widetilde{P}$
of vertices with minimum $(i+1)$-st derivative,
which can also be computed.
These vertices will be separated from the stationary vertices
 if their $(i+1)$-st derivatives are negative.

Therefore, by recursively applying the arguments in the above two cases,
the procedure $\mathfrak{D}^{(i+1)}$ can correctly return the $(i+1)$-st derivatives for all vertices.

\vspace{5pt}

\noindent \textbf{Relating the norm $\norm{f^{(i+1)}}_\omega$
with active edges.}
Observe that each $X \in \mathcal{C}[\sigma_{i+1}]$
is recovered by the procedure in some (least) densest subset problem.
We use $\delta^{(i)}(X)$ to denote its density at the moment that it is
recovered.
Then, it follows that 
for each $u \in X$, $f^{(i+1)}(X) = f^{(i+1)}_u = \delta^{(i)}(X)$,
and
$\omega(X) \cdot f^{(i+1)}(X) = \sum_{e \in I^{(i+1)}_X} c^{(i)}_e
- \sum_{e \in S^{(i+1)}_X} c^{(i)}_e$,
where $c^{(i)}_e = w_e \Delta^{(i)}_e$.
Observe that we use the superscript $(i+1)$ for the edge sets
$I^{(i+1)}_X$ and $S^{(i+1)}_X$,
because $X$ could actually be returned in a recursive instance of the
(least) densest subset problem.  Hence, we have the following:

\begin{align}
\norm{f^{(i+1)}}^2_\omega = \sum_{u \in N} \omega_u \cdot (f^{(i+1)}_u)^2
& = \sum_{X \in \mathcal{C}[\sigma_{i+1}]: f^{(i+1)}(X) \neq 0 } \omega(X) \cdot \delta^{(i)}(X) \cdot f^{(i+1)}(X) \\
& = \sum_{X \in \mathcal{C}[\sigma_{i+1}]: f^{(i+1)}(X) \neq 0 } 
(\sum_{e \in I^{(i+1)}_X} c^{(i)}_e
- \sum_{e \in S^{(i+1)}_X} c^{(i)}_e)
\cdot f^{(i+1)}(X)  \label{eq:2nd}
\end{align}

Fixing an edge $e \in E^{(i)}_+$,
we next consider the coefficient of $c^{(i)}_e =   w_e \Delta^{(i)}_e$
in (\ref{eq:2nd}).
Observe there is a contribution of $f^{(i+1)}(X)$,
where $X \in \mathcal{C}[\sigma_{i+1}]$ is
such that $\arg \min_{u \in \head_e} [u]_{\sigma_{i+1}} \subseteq X$;
and there is a contribution of $-f^{(i+1)}(Y)$,
where $Y \in \mathcal{C}[\sigma_{i+1}]$ is
such that $\arg \max_{u \in \tail_e} [u]_{\sigma_{i+1}} \subseteq Y$.
Therefore, the coefficient of $c^{(i)}_e$
is $f^{(i+1)}(\arg \min_{u \in \head_e} [u]_{\sigma_{i+1}}) -
f^{(i+1)}(\arg \max_{u \in \tail_e} [u]_{\sigma_{i+1}} ) = - \Delta^{(i+1)}_e$.

Therefore, (\ref{eq:2nd}) equals $- \sum_{e \in E^{(i)}_+} w_e \Delta^{(i)}_e \Delta^{(i+1)}_e$, as required.
\end{proof}

So far, we have shown that derivatives of $f$, i.e., $f^{(i)}$ for $i = 1,2,\ldots$, exist and can be found by the above procedures.
Referring to Definition~\ref{defn:rules}, the following lemma shows that once we have computed $f^{(1)}$, there is a way to compute $\vp^{(1)}_u(e)$, which might not be unique.

\begin{lemma}[Existence of $\vp^{(1)}_u(e)$]
\label{lemma:existence_of_varphi}
	Given $f^{(1)}$ and $\sigma_1 = \sigma_0(f^{(1)})$, it is possible to assign for each $e\in E$ and $u\in \tail_e\cup \head_e$, the value $\vp^{(1)}_u(e)$ such that $\omega_u f^{(1)}_u = \vp^{(1)}_u = \sum_{e\in E:u\in \tail_e\cup \head_e} \vp^{(1)}_u(e)$ for $u\in N$.
\end{lemma}

\begin{proof}
	Note that we only need to consider $e\in E^{(0)}_+$, i.e., $c^{(0)}_e > 0$,
	as otherwise we have $\vp^{(1)}_u(e) = 0$ for all $u\in \tail_e\cup\head_e$.
	Consider each equivalence class $P \in \mathcal{C}[\sigma_1]$. 
	Let $\delta_M$ be the common $f^{(1)}$ value of vertices in $P$. Recall that $\delta_M = \delta^{(0)}(P)$ if $\delta_M > 0$ and $\delta_M = \widetilde{\delta}^{(0)}(P)$ if $\delta_M \leq 0$, when $P$ is separated. For notational convenience, we drop the superscripts when there is no ambiguity.
	
	Suppose $P\cap T = \emptyset$.
	Consider any configuration $\vp^{(1)}$ in which edge $e\in I_P$ supplies a measure rate of $c_e = w_e \Delta_e$ to the vertices of $P$ and edge $e\in S_P$ demands a measure rate of $c_e$ from the vertices of $P$.
	Each vertex $u\in P$ is supposed to gather a net rate of $\omega_u\cdot \delta_M$, where we call any deviation the \emph{surplus} or \emph{deficit}.
	
	Define a directed graph with vertices $P$, such that there is an arc from $u$ to $v$ if non-zero measure rate can be transfered from $u$ to $v$.
	Note that if there is a directed path from $u$ with non-zero surplus to $v$ with non-zero deficit, then we can decrease the surplus of $u$ and the deficit of $v$. 
	Suppose all surpluses are zero (which implies zero deficits), then $\vp^{(1)}_u(e)$ is the measure rate of $u$ supplied by $e$ (if $u\in I_e$) or demanded by $e$ (if $u\in S_e$).
	
	Suppose there exists $u\in P$ with non-zero surplus and all vertices $P'\subseteq P$ reachable from $u$ have zero deficits.
	Then, it follows that $\delta(P') > \delta_M$, which is a contradiction.
	Hence, $\vp^{(1)}_u(e)$ exists and can be found by a standard flow problem between edges $I_P\cup S_P$ and vertices $P$.
	
	Next, consider the case when $U$ contains a stationary vertex $v_0$.
	Note that in this case we have $\delta_M = 0$, and there does not exist any $X\subseteq P$ with $\delta(X)>0$ or $\widetilde{\delta}(X)<0$.
	Each vertex $u\in P\cap N$ is supposed to gather a net rate of $0$, while the stationary vertex $v_0$ is supposed to gather a net rate of $c(I_P) - c(S_P)$.   Hence,
	we can still define the notion of surplus and deficit based on the deviation from the net rate
	a vertex is supposed to receive.
		As before, we define a directed graph with vertices $P$, and try to decrease any surplus until it is not possible.
		
	Suppose there is a vertex with non-zero surplus and all its reachable vertices $P'$ have zero deficits.
	If $P'\cap T = \emptyset$, then we have $\delta(P')>0$, which is a contradiction; otherwise we have $v_0\in P'$.
	Then, we know that the net rate of vertices in $P'$ is larger than $c(I_P)-c(S_P)$, which implies that $X = P\backslash P'$ (which contains no stationary vertex) has $\widetilde{\delta}(X)<0$ and it is also
	a contradiction.
\end{proof}

\subsection{Properties of Diffusion Operator}

Recall that the diffusion process defined in
Definition~\ref{defn:rules} induces
an operator $\L_\omega: \R^V \ra \R^N$
defined by $\L_\omega f := - \Pi_N \frac{df}{dt} = - f^{(1)}_N \in \R^N$.
We next prove some properties of this operator.

\begin{lemma}[Properties of Diffusion Operator]
\label{lemma:first_deriv}
Consider the diffusion process in Definition~\ref{defn:rules}.
Then, we have the following:
\begin{compactitem}
\item[1.] $\frac{d \norm{f}^2_\omega}{dt} = - 2 \langle f, \L_\omega f \rangle_\omega$.
\item[2.] $\frac{d \Q(f)}{dt} = - \norm{\L_\omega f}^2_\omega = - \norm{f^{(1)}_N}^2_\omega$.
\end{compactitem}
\end{lemma}

\begin{proof}
The first statement follows from a standard calculus computation 
$\frac{d \norm{f}^2_\omega}{dt} = \frac{d}{d t} \langle f_N, f_N \rangle_\omega
= 2 \langle f_N, \frac{d f_N}{d t} \rangle_\omega
= - 2 \langle f, \L_\omega f \rangle_\omega$,
recalling that the arguments of the inner-product $\langle \cdot, \cdot \rangle_\omega$
can come from either $\R^V$ or $\R^N$.

For the second statement,
recall that $\Q(f) = \frac{1}{2} \sum_{e \in E^{(0)}_+} w_e (\Delta^{(0)}_e)^2$,
where $\Delta^{(0)}_e := f(\arg \max_{u \in \tail_e} [u]_{\sigma_0})
- f(\arg \min_{u \in \head_e} [u]_{\sigma_0})$
and $\sigma_0 = \varsigma(f)$ is the ordered equivalence relation induced by $f$.

By the Envelope Theorem, we have
$$\frac{d \Q(f)}{dt} = \sum_{e \in E^{(0)}_+} w_e \Delta^{(0)}_e \cdot 
(f^{(1)}(\arg \max_{u \in \tail_e} [u]_{\sigma_1})
- f^{(1)}(\arg \min_{u \in \head_e} [u]_{\sigma_1}))
= \sum_{e \in E^{(0)}_+} w_e \Delta^{(0)}_e \cdot \Delta^{(1)}_e
= - \norm{f^{(1)}}^2_\omega,$$
where the last equality comes from Lemma~\ref{lemma:deriv}.
\end{proof}

%% file: laplacian.tex
\section{Spectral Analysis of Directed Hyperedge Expansion}
\label{sec:spectral}

Recall that the discrepancy ratio in Definition~\ref{defn:disc_ratio}
is used to study directed hyperedge expansion (defined in Section~\ref{sec:intro}),
in which case the diffusion process has no stationary vertex, i.e., $V = N$.
The main result of this section is as follows.

\begin{lemma}[Non-Trivial Eigenvalue of Diffusion Operator]
\label{lemma:eigenvalue}
The parameter $\gamma_2 := \min_{\vec{0} \neq f \perp_\omega \vec{1}} \D(f)$
is an eigenvalue of $\L_\omega$, where any minimizer is a corresponding eigenvector.
\end{lemma}

\begin{proof}
After establishing
Lemma~\ref{lemma:change_ratio} below,
the proof is exactly the same as~Theorem 4.1 in \xcite{chan2016spectral}.
\end{proof}

Using standard techniques in~\xcite{chan2016spectral, yoshida2016nonlinear},
we can obtain the following Cheeger Inequality.

\begin{claim}[Cheeger Inequality for Directed Hyperedge Expansion]
\label{claim:cheeger}
Given a directed hypergraph $H$, its directed hyperedge expansion satisfies the
following:
$$\frac{\gamma_2}{2} \leq \phi_H \leq 2 \sqrt{\gamma_2}.$$
\end{claim}

Since the proof of Claim~\ref{claim:cheeger}
is very similar to previous results on undirected hypergraphs~\xcite{chan2016spectral}
and normal directed graphs~\xcite{yoshida2016nonlinear},
we will not repeat the whole proof here, and instead highlight the differences
for proving Lemma~\ref{lemma:eigenvalue}.
We remark that the lower bound $\gamma_2 \leq 2 \phi_H$ is proved
by showing that it is possible to derive a desirable vector given
a subset with small expansion.

The following result follows from essentially the same
proof as in~Lemma~4.2 in \xcite{chan2016spectral}.

\begin{proposition}[Discrepancy Ratio Coincides with Rayleigh Quotient]
\label{prop:disc_ray}
For $V = N$, suppose $\L_\omega$ is the diffusion operator
given by Definition~\ref{defn:rules}.
Then, for all $f \in \R^V$,  $\langle f, \L_\omega f \rangle_\omega = 2 \Q(f)$.
This implies that the discrepancy ratio $\D(f) = \frac{\langle f, \L_\omega f \rangle_\omega}{\langle f,  f \rangle_\omega}$.
\end{proposition}
\begin{proof}
	We have:
	\begin{align*}
	\sum_{e\in E_+^{(0)}} w_e(\Delta_e(f))^2 &= \sum_{e\in E_+^{(0)}} w_e (\max_{u\in \Te} f_u - \min_{v\in\He} f_v)^2
	= \sum_{e\in E_+^{(0)}} \sum_{(u,v)\in S_e\times I_e}a^e_{uv} (f_u - f_v)^2 \\
	& = \sum_{u,v \in V: f_u > f_v}\sum_{e\in E_+^{(0)}} a^e_{uv} (f_u - f_v)^2
	= \sum_{u,v \in V: f_u > f_v} a_{uv} (f_u - f_v)^2 \\
	& = \sum_{u\in V}(\sum_{v\in V: v\neq u}a_{uv})f_u^2 - \sum_{u\in V}\sum_{v \in V: v\neq u} a_{uv} f_u f_v  \\
	& = \sum_{u\in V}(\omega_u-a_{uu})f^2_u - \sum_{u\in V}\sum_{v \in V: v\neq u} a_{uv} f_u f_v \\
	&= f^\T (\W - A_f)f  = \langle f,\L_\omega f\rangle_\omega,
	\end{align*}
	which yields the claim by definition of $\Q(f)$. 
\end{proof}

\begin{lemma}[Simple Spectral Properties of Diffusion Operator]
\label{lemma:simple_spectral}
Suppose $\vec{1} \in \R^V$ is a vector such that every coordinate has value 1.
Then, we have the following.
\begin{compactitem}
\item[1.] $\L_\omega \vec{1} = 0$.
\item[2.] Suppose $V = N$. For all $f \in \R^V$, $\L_\omega f \perp_\omega \vec{1}$.
\end{compactitem}
\end{lemma}

\begin{proof}
Recall that the diffusion process is defined such that
$\frac{df}{dt} = - \L_\omega f$.

For the first statement, if every vertex has the same density (e.g., every
vertex $u$ has $f_u = 1$),
then according to Definition~\ref{defn:rules},
there will be no measure flowing between vertices.
Therefore, $\frac{df}{dt} = 0$, i.e., $\L_\omega \vec{1} = 0$.

For the second statement, observe that if there
is no stationary vertex, then the total measure is preserved.
Therefore, we have $0 = \sum_{u \in V} \omega_u \frac{d f_u}{d t}
= - \langle \vec{1}, \L_\omega f \rangle_\omega$, as required.
\end{proof}

\begin{lemma}[Rate of Change of Discrepancy Ratio]
\label{lemma:change_ratio}
Consider the diffusion process in Definition~\ref{defn:rules}
for the case $V=N$. Then, we have

$$\frac{d \D(f)}{dt} = - \frac{2}{\norm{f}^4_\omega}
\cdot (\norm{f}_\omega^2 \cdot \norm{\L_\omega f}^2_\omega
- \langle f , \L_\omega f  \rangle^2_\omega)
\leq 0,$$
where the inequality follows from the Cauchy-Schwarz inequality
and equality holds \emph{iff} $\L_\omega f \in \mathrm{span}(f)$.
\end{lemma}

\begin{proof}
The proof is essentially the same as~Lemma~4.10 in \xcite{chan2016spectral}.  We have

$\frac{d \D(f)}{dt} = 2 \frac{d}{d t} \frac{\Q(f)}{\norm{f}^2_\omega}
= \frac{2}{\norm{f}^4_\omega} (\norm{f}^2_\omega \cdot \frac{d \Q(f)}{dt}
- \Q(f) \cdot \frac{d  \norm{f}^2_\omega}{dt}) = 
- \frac{2}{\norm{f}^4_\omega}
\cdot (\norm{f}_\omega^2 \cdot \norm{\L_\omega f}^2_\omega
- \langle f , \L_\omega f  \rangle^2_\omega)$,
where the last equality follows from Lemma~\ref{lemma:first_deriv}
and Proposition~\ref{prop:disc_ray}.
\end{proof}

%% file: subgradient.tex
\section{Subgradient of Quadratic Form}
\label{sec:subgradient}

Recall that since the
$f$ values of stationary vertices do not change,
in this section, we treat $\Q : \R^N \ra \R$
as a function of $f_N \in \R^N$. In the quadratic form $\Q$,
there is no particular notion for the weight
$\omega_u$ of a non-stationary vertex $u \in N$.
In fact, we could choose arbitrary positive weight for a non-stationary vertex,
and this will affect the trajectory of the diffusion process.
However, for simplicity, we consider
the special case when every non-stationary vertex has unit weight.

\noindent \textbf{Subgradient Method.}  Application~2 defined
in Section~\ref{sec:prelim} aims to find a subgradient of $\Q$
that is also a descent direction.  Although the subgradient method~\xcite{Shor1985}
only requires a subgradient that might not necessarily be a descent direction,
we think it is interesting to make the connection that
the diffusion process can be viewed as a continuous subgradient method such that
at any moment, the update direction $\Pi_N \frac{df}{dt} = - \L_\omega f$ is a descent direction
and $\L_\omega f \in \R^N$ is a subgradient of $\Q$.  The following is the
main result of this section.

\begin{lemma}[Diffusion Operator Gives Subgradient]
\label{lemma:subgradient}
Given a directed hypergraph $H$ such that each non-stationary vertex has unit weight,
the diffusion operator $\L_\omega: \R^V \ra \R^N$ resulting from 
Definition~\ref{defn:rules} gives a subgradient for $\Q: \R^N \ra \R$.

Recall that Lemma~\ref{lemma:first_deriv} implies that $\Pi_N \frac{df}{dt} = - \L_\omega f$
gives a descent direction, because in the diffusion process, $\frac{d\Q}{dt} = - \norm{\L_\omega f}^2_\omega \leq 0$.
\end{lemma}

\begin{corollary}[Convergence of Diffusion Process]
\label{cor:convergence}
In the diffusion process in Definition~\ref{defn:rules} where each non-stationary vertex
has unit weight,
the value of the potential function $\Q$ converges to its minimum value.
\end{corollary}

\begin{proof}
The result follows readily from Shor's subgradient method~\xcite{Shor1985},
once we establish in Lemma~\ref{lemma:subgradient} that
$\L_\omega$ returns a subgradient.  However, we still give some intuition here.

In Lemma~\ref{lemma:first_deriv},
we showed that
$\frac{d\mathsf{Q}}{dt} = -\norm{\L_\omega f}_\omega^2$. Using the fact that $\mathsf{Q}$ is non-increasing and lower-bounded by $0$, we conclude that $\mathsf{Q}$ converges to some value $\mathsf{Q}^*$ in the diffusion process.
This immediately implies that $\norm{\L_\omega f}_\omega^2$ converges to $0$ as $t \to \infty$.

We next consider the trajectory $\{f(t) \in \R^V\}_{t}$. Since our diffusion process cannot increase $\max_u \{f_u\}$ or decrease $\min_u \{f_u\}$, the 
trajectory falls within a compact set in $\R^V$.  Therefore,
there exists a subsequence in the trajectory that converges to some $f^* \in \R^V$.
Hence, it follows that $\L_\omega f^* = 0$.  Since $\L_\omega$ gives a
subgradient of $\Q$, it follows that $f^*$ is a minimum of $\Q$.
\end{proof}

\vspace{5pt}

\noindent \textbf{Permutation as Ordered Equivalence Relation.}
We next consider a special class of ordered equivalence relation,
in which each equivalence class is a singleton.
In other words, each such ordered equivalence relation
is a permutation of $V$.
We denote by $\mathcal{S}(V)$ the set of equivalence relations that are permutations of~$V$.

\begin{definition}[Quadratic Form Induced by a Permutation]
Given a permutation $\sigma \in \mathcal{S}(V)$ on $V$,
we consider the quadratic form $\Q_\sigma: \R^N \ra \R$
restricted to $f_N \in \R^N$ that is consistent with $\sigma$,
i.e., $\sigma \sqsupseteq \varsigma(f)$ is a refinement
of the equivalence relation induced by $f$.

The interpretation is that 
we use permutation $\sigma$ to resolve ties between
vertices having the same $f$ values.
Hence, we can interpret $\Q_\sigma$ as a function on $f_N \in \R^N$
such that the $f$ values of non-stationary vertices are distinct and consistent with $\sigma$.
\end{definition}

\noindent \textbf{Active Edges.}
Given $\sigma \in \mathcal{S}(V)$,
define $E(\sigma) := \{e \in E: \max_{u \in \tail_e} [u]_\sigma \succ 
\min_{v \in \head_e} [v]_\sigma\}$ to be the set of edges that are active
with respect to $\sigma$.


\noindent \textbf{Gradient of $\Q_\sigma$.}  Observe that
$\Q_\sigma(f_N) = \sum_{e \in E(\sigma)} w_e \Delta_e(f)^2$
is differentiable for $f_N \in \R^N$ such that
$\varsigma(f) \sqsubseteq \sigma$ and $f_N$ has distinct coordinates.
Hence, its gradient $\nabla \Q_\sigma \in \R^N$ is well-defined.

Since we consider a permutation $\sigma$,
the active edges $E(\sigma)$ corresponds naturally to an edge-weighted
normal undirected graph.
Specifically, each $e \in E(\sigma)$ contributes its weight
$w_e$ to the pair $\{\max_{u \in \tail_e} [u]_\sigma, \min_{v \in \head_e} [v]_\sigma\}$.
The matrix $A_\sigma \in \R^{V \times V}$ gives the weight of each pair,
where the diagonal entries are chosen such that each row corresponding to $u \in N$
sums to $\omega_u = 1$; the rows corresponding to stationary vertices are not important.
Then, it can be checked that 
$\nabla \Q_\sigma(f_N) = \Pi_N (\I - A_\sigma) f$, where the values $f_T \in \R^T$ for
stationary vertices are fixed constants.

\begin{lemma}[$\nabla \Q_\sigma$ Gives Subgradient]
\label{lemma:nabla_subgradient}
For $f \in \R^V$, suppose $\sigma \in \mathcal{S}(V)$
is consistent with $f$, i.e., $\varsigma(f) \sqsubseteq \sigma$.  Then,
$\nabla \Q_\sigma(f_N)$ is a subgradient of $\Q$ at $f_N \in \R^N$.
\end{lemma}

\begin{proof}
Observe that if the coordinates of $f_N \in \R^N$ are distinct,
then $\Q(f_N) = \Q_\sigma(f_N)$ in a neighborhood of $f_N$.
Hence, the gradient $\nabla \Q_\sigma(f_N)$ is the only subgradient of $\Q$
at $f_N$.

Otherwise, we imagine that the coordinates
of $f_N$ are perturbed slightly to form $\widehat{f}_N$ with distinct coordinates
that is consistent with $\sigma$.
Moreover,
we consider $f_N(\theta) := f_N + \theta (\widehat{f}_N - f_N)$
for $\theta \in [0,1]$.

If we choose $\widehat{f}_N$ close enough to $f_N$,
we may assume that for all $\theta \in (0,1]$,
$f_N(\theta)$ has distinct coordinates and is consistent with
$\sigma$.  Since $\Q(f_N(\theta)) = \Q_\sigma(f_N(\theta))$ is differentiable
at $f_N(\theta)$ and $\Q$ is convex,
we have for any $g_N \in \R^N$,

$\Q(g_N) \geq \Q(f_N(\theta)) + \langle g_N - f_N(\theta), \nabla \Q_\sigma(f_N(\theta)) \rangle$.

Observing that $\Q$ and $\nabla \Q_\sigma$ are both continuous, as $\theta$ tends to $0$,
we have
$\Q(g_N) \geq \Q(f_N) + \langle g_N - f_N, \nabla \Q_\sigma(f_N) \rangle$,
which shows that $\nabla \Q_\sigma(f_N)$ is a subgradient of $\Q$ at $f_N$.
\end{proof}

Since a convex combination of subgradients is also a subgradient,
the following lemma combined with Lemma~\ref{lemma:nabla_subgradient}
proves Lemma~\ref{lemma:subgradient}.  We give its proof in 
Section~\ref{sec:diffusion_subgradient}.

\begin{lemma}[Diffusion Operator as Subgradient]
\label{lemma:diffusion_subgradient}
Given $f \in \R^V$, there exists a distribution~$\mathcal{D}$ on permutations $\sigma \in \mathcal{S}(V)$ that are consistent with $f$, i.e., $\varsigma(f) \sqsubseteq \sigma$
such that the diffusion operator
satisfies $\L_\omega f = \expct_{\sigma \leftarrow \mathcal{D}}[\nabla \Q_\sigma(f_N)]$.
\end{lemma}

\subsection{Diffusion Operator as Subgradient: Proof of Lemma~\ref{lemma:diffusion_subgradient}}
\label{sec:diffusion_subgradient}

For $e \in E$ and an ordered equivalence relation $\sigma$,
we denote $S^\sigma_e := \arg \max_{u \in \tail_e} [u]_\sigma \subseteq \tail_e$
and $I^\sigma_e := \arg \min_{u \in \head_e} [u]_\sigma \subseteq \head_e$.

\noindent \textbf{Proof Outline.}  Recall that
$\L_\omega f$ is defined in terms of the first derivative $f^{(1)}$
in Lemma~\ref{lemma:deriv} for the case $i=0$; hence, for simplicity,
we drop the superscript $(i)$ when there is no ambiguity in the following discussion.
The procedure $\mathfrak{D}^{(1)}$ essentially finds
each equivalence class $P$ of $\sigma_1 = \sigma_0(f^{(1)})$ (recall that $\sigma_0 = \varsigma(f)$),
where each such $P$ is associated with its sets $I_P$ and $S_P$ of edges.
When $P$ does not contain any stationary vertex,
the first derivative is formed by allocating the contribution $c(I_P) - c(S_P)$
uniformly among the vertices in $P$, where for $e \in E_+$,
$c_e := w_e \Delta_e$.
Observe that each permutation $\sigma \in \mathcal{S}(V)$
will also induce an allocation: for each $e \in E_+$,
it will contribute $c_e$ to the unique vertex in $I^\sigma_e$
and $-c_e$ to the unique vertex in $S^\sigma_e$.
Moreover, in this case, the contribution received by
a non-stationary vertex $v$ is exactly $\frac{\partial \Q_\sigma (f_N)}{\partial f_v}$.

From the way each equivalence class $P$ is extracted by $\mathfrak{D}^{(1)}$,
it suffices to consider permutations~$\sigma$ that are refinements
of~$\sigma_1$. Specifically, we shall prove that
for each equivalence class~$P$, there is some distribution of permutations of $P$
that will induce an expected allocation of contributions from edges in~$I_P$ and~$S_P$ which
is equivalent to uniform allocation among vertices in $P$.  For an equivalence
class~$P$ that contains a stationary vertex, we need to be more careful in the analysis,
because each non-stationary vertex in~$P$ should receive zero net contribution,
while the stationary vertex will receive a net contribution of $c(I_P) - c(S_P)$.

\noindent \textbf{Distribution of Permutations as Vectors in $[0,1]^{\mathcal{S}(V)}$.}
For an ordered equivalence relation $\sigma$, we denote $\mathcal{S}(\sigma)$
as the set of permutations that are refinements of $\sigma$.
We define $\vec{\sigma} \in [0,1]^{\mathcal{S}(V)}$
to be a vector that satisfies the following properties:
\begin{compactitem}
\item[(a)] for each $\tau \in \mathcal{S}(\sigma)$,
the coordinate of $\vec{\sigma}$ corresponding to $\tau$ is $\frac{1}{|\mathcal{S}(\sigma)|}$,

\item[(b)] for $\tau \notin \mathcal{S}(\sigma)$,
the corresponding coordinate of $\vec{\sigma}$ is $0$.
\end{compactitem}

\noindent \textbf{General Strategy to Find Desired Distribution of
Permutations.}  We shall define an appropriate function  $\sigma: [0,1] \ra \mathcal{O}(V)$
such that for $0 \leq t_1 < t_2 \leq 1$, we have the
monotone property $\sigma(t_1) \sqsubseteq \sigma(t_2)$.
We shall imagine that the parameter $t \in [0,1]$ corresponds to time.

%

In other words, as $t$ increases from $0$ to $1$,
we can possibly have a more refined ordered equivalence relation.
As we shall see, the desired distribution of permutations
will be given by $\alpha := \int_{0}^1 \vec{\sigma}(t) dt$.

\noindent \textbf{Concatenation of Ordered Equivalence Relations.}
As discussed above, we shall consider each equivalence class $P$
of $\sigma_1$ separately.  Hence, we need
some notation to concatenate equivalence relations on disjoint subsets.
For disjoint subsets $P_1$ and $P_2$ of $V$,
given $\tau_1 \in \mathcal{O}(P_1)$
and $\tau_2 \in \mathcal{O}(P_2)$,
we define the concatenation $\tau_1 \circ \tau_2 \in \mathcal{O}(P_1 \cup P_2)$
naturally
as the ordered equivalence relation on $P_1 \cup P_2$ in which 
$\mathcal{C}[\tau_1 \circ \tau_2] = \mathcal{C}[\tau_1] \cup \mathcal{C}[\tau_2]$
and every equivalence class in $\tau_1$ is lower than every one in $\tau_2$.

%

Hence, for each $t \in [0,1]$, once we define a refinement $\tau_j(t) \in \mathcal{O}(P_j)$ within each equivalence class $P_j$, where
$P_1 \prec P_2 \prec \ldots \prec P_J$ are the equivalence classes of $\sigma_1$,
the desired ordered equivalence relation on $V$ is
${\sigma}(t) := {\tau}_1(t) \circ \cdots \circ \tau_J(t) \in \mathcal{O}(V)$.

\noindent \textbf{Target Contribution Received by Each Vertex.} Recall that when an equivalence class $P$ is extracted by $\mathfrak{D}^{(1)}$,
there are associated edges $I_P$ and $S_P$,
where each edge $e \in I_P \cup S_P$ has an associated contribution~$c_e$.
If $e \in I_P$, a positive contribution~$c_e$ is allocated to vertices in $P$,
while if $e \in S_P$, a negative contribution $- c_e$ is allocated.
The contribution allocated to each vertex $v \in P$ is determined by the following cases.

\begin{compactitem}
\item[(a)] If $P$ does not contain any stationary vertex,
then each $v \in P$ is supposed to receive a net contribution
of $C_v := \frac{\sum_{e \in I_P} c_e  - \sum_{e \in S_P} c_e }{|P|}$.
\item[(b)] Otherwise, we may assume that there is exactly one stationary vertex
$v_0$ in $P$, which receives $C_{v_0} := \sum_{e \in I_P} c_e  - \sum_{e \in S_P} c_e$,
while all other non-stationary vertex $v$ receives $C_v := 0$.
\end{compactitem}

\noindent \textbf{Contribution Induced by a Permutation.}
Given a permutation $\pi \in \mathcal{S}(P)$,
the contribution received by a vertex $v \in P$
is $c_v(\pi) := \sum_{e \in I_P: v = I^\pi_e} c_e - 
\sum_{e \in S_P: v = S^\pi_e} c_e$.
Recall that our goal is to show that
there exists a distribution $\alpha \in [0,1]^{\mathcal{S}(P)}$
of permutations
such that for all $v \in P$,
$C_v = \sum_{\pi \in \mathcal{S}(P)} \alpha(\pi) \cdot c_v(\pi)$.

\noindent \textbf{Defining $\tau: [0,1] \ra \mathcal{O}(P)$ for 
Each Equivalence Class $P$ of $\sigma_1$.}
Our strategy is to define $\tau: [0,1] \ra \mathcal{O}(P)$,
where $\tau(0)$ is the trivial equivalence relation having $P$
as the only equivalence class, and as $t$ increases, $\tau(t)$ can become more refined.
Recall that each $\tau(t) \in \mathcal{O}(P)$
corresponds to a distribution $\vec{\tau}(t) \in [0,1]^{\mathcal{S}(P)}$
of permutations on $P$.

\noindent \textbf{Contribution Received by a Vertex at Time $t$.}
Suppose at time $t \in [0,1]$,
we have some $\tau(t) \in \mathcal{O}(P)$.
Then, at this moment,
the contribution received by vertex $v \in P$
is $c^\tau_v(t) := \sum_{\pi \in \mathcal{S}(P)}  \vec{\tau}(t)(\pi) \cdot c_v(\pi)$.
Hence, it suffices to show that
it is possible to define $\tau: [0,1] \ra \mathcal{O}(P)$
such that for all $v \in P$,
$C_v = \int_{0}^1 c^\tau_v(t) dt$.

For a subset $X \subseteq P$,
we denote $C_X := \sum_{v \in X} C_v$ as
the target contribution received by vertices in $X$,
$c^\tau_X[0..t] := \sum_{v \in X} \int_{0}^t c^\tau_v(s) ds$
as the contribution received up to time $t$,
and $c^\tau_X[t..1] := C_X - c^\tau_X[0..t]$
as the contribution that is supposed to be received after time $t$.

Given $\tau \in \mathcal{O}(P)$
and a subset $X$ of some equivalence class in $\mathcal{C}[\tau]$,
we denote $I^\tau_X := \{e \in I_P: I^\tau_e \subseteq X\}$,
and $S^\tau_X := \{e \in S_P: S^\tau_e \cap X \neq \emptyset\}$.

\noindent \textbf{Invariants for Updating Ordered Equivalence Relation $\tau$.}
At time $t=0$, we start with
the trivial $\tau(0) \in \mathcal{O}(P)$ in which there is only one equivalence class $P$.
When $t$ increases from $0$ to $1$, we shall refine the ordered equivalence relation at certain times such that the following invariants hold.
Suppose for some $t \in [0,1]$, we have already determined $\tau$ on $[0..t)$.
Then, we have the following.

\begin{compactitem}
\item[\textsf{(I1)}]
For each equivalence class $X \in \mathcal{C}[\tau(t)]$,
$c^\tau_X[t..1] = (1 - t) \cdot (c(I^{\tau(t)}_X) - c(S^{\tau(t)}_X))$.

\item[\textsf{(I2)}]
For a proper subset $Y \subset X \in \mathcal{C}[\tau(t)]$,

let $\hat{c}^\tau_Y[t..1] = (1-t) 
\cdot (\sum_{e \in I^{\tau(t)}_X: I^{\tau(t)}_e \cap Y \neq \emptyset} c_e
- \sum_{e \in S^{\tau(t)}_X: S^{\tau(t)}_e \subseteq Y} c_e)$
denote the maximum contribution that could be received by $Y$ in $[t,1]$,
assuming that $\tau$ is refined at time $t$ such that all vertices in $Y$
are the minimum in the ordering among $X$.

Similarly, let
$\check{c}^\tau_Y[t..1] = (1-t) 
\cdot (\sum_{e \in I^{\tau(t)}_X: I^{\tau(t)}_e \subseteq Y} c_e
- \sum_{e \in S^{\tau(t)}_X: S^{\tau(t)}_e \cap Y \neq \emptyset} c_e)$
denote the minimum contribution that could be received by $Y$ in $[t,1]$,
assuming that $\tau$ is refined at time $t$ such that all vertices in $Y$
are the maximum in the ordering among $X$.

Then, we require that
\begin{equation} \label{eq:bounds}
\check{c}^\tau_Y[t..1] \leq {c}^\tau_Y[t..1] \leq \hat{c}^\tau_Y[t..1].
\end{equation}
\end{compactitem}

\noindent \textbf{Initial Setup.}
For $t=0$, the invariants
\textsf{(I1)} and \textsf{(I2)}
both hold, because (by Lemma~\ref{lemma:existence_of_varphi}) each edge $e \in I_P \cup S_P$
can allocate their contribution appropriately such that
every vertex receives its target contribution.

\noindent \textbf{Crucial Moment.}
As $t$ increases, as long as $\tau(t)$ can only become more refined,
invariant~\textsf{(I1)} will continue to hold.
However, at some crucial moment, say $t_0$, invariant~\textsf{(I2)}
might be about to be violated. 
We use $t^-$ to denote the moment just before $t_0$.
Suppose one of the inequalities in~(\ref{eq:bounds})
becomes tight for some $Y \subset X \in \mathcal{C}[\tau(t^-)]$.

For the case (i) $\check{c}^\tau_Y[t^-..1] = {c}^\tau_Y[t^-..1]$,
the vertices in $Y$ need to receive as little contribution as
they could possibly get.  Hence, at this moment $t_0$,
the equivalence class $X$ should be separated into $X \setminus Y \prec Y$.

Similarly, for the case (ii) 
${c}^\tau_Y[t^-..1] = \hat{c}^\tau_Y[t^-..1]$.
At the moment $t_0$, the equivalence class should be
separated into $Y \prec X \setminus Y$.

We need to check that the invariants still hold after the refinement.
Note that in case~(ii), when ${c}^\tau_Y[t^-..1] = \hat{c}^\tau_Y[t^-..1]$, we must have $\check{c}^\tau_{X \setminus Y}[t^-..1] = {c}^\tau_{X \setminus Y}[t^-..1]$ (case~(i) for $X \setminus Y$), since ${c}^\tau_Y[t^-..1] +  {c}^\tau_{X \setminus Y}[t^-..1] =  {c}^\tau_{X}[t^-..1] = \hat{c}^\tau_Y[t^-..1] + \check{c}^\tau_{X \setminus Y}[t^-..1]$.

Hence we only need to consider case~(i).

For invariant~\textsf{(I1)},
observe that invariant~\textsf{(I1)}
holds for~$X$.  Hence, after splitting $X$ into $X \setminus Y \prec Y$.
The equality $\check{c}^\tau_Y[t^-..1] = {c}^\tau_Y[t^-..1]$
implies that invariant~\textsf{(I1)} must hold for both $Y$ and $X \setminus Y$.

For invariant~\textsf{(I2)}, suppose for contradiction's sake,
the invariant is violated by
some proper subset $A \subset Y$.  Since invariant~\textsf{(I2)}
holds for $Y$, it follows that either $c^\tau_A[t_0..1]$ or $c^\tau_{Y \setminus A}[t_0..1]$
is too large.  Without loss of generality, we assume that
$c^\tau_A[t_0..1] > \hat{c}^\tau_A[t_0..1]$.
Since invariant~\textsf{(I2)} holds for $X \setminus Y$,
it follows that at time $t^-$, 
$c^\tau_{X \setminus Y}[t^-..1] = \hat{c}^\tau_{X \setminus Y}[t^-..1]$.

Therefore, if we consider $B := A \cup (X \setminus Y)$,
we have $c^\tau_B[t^-..1] = c^\tau_A[t_0..1] + c^\tau_{X \setminus Y}[t^-..1]
> \hat{c}^\tau_A[t_0..1] + \hat{c}^\tau_{X \setminus Y}[t^-..1]
\geq \hat{c}^\tau_B[t^-..1]$, violating
the assumption that $t_0$ is the earliest time such that 
an inequality in~(\ref{eq:bounds}) becomes tight.

Observe that at a certain moment $t_0$,
we might have to repeat this procedure several times because
the inequalities in~(\ref{eq:bounds})
can be tight for proper subsets in several equivalence classes in $\tau(t^-)$.

\noindent \textbf{Terminating Condition.}
As $t=1$, the procedure terminates.  Observe that
invariants~\textsf{(I1)}
and~\textsf{(I2)} imply that for all $v \in P$,
$C_v = c^\tau_v[0..1]$, which means that every vertex has received
its target contribution, as required.

The above procedure is summarized in Algorithm~\ref{alg:permutation_disribution}.

\begin{algorithm}[htb]
	\caption{Constructing $\tau: [0,1] \ra \mathcal{O}(P)$}\label{alg:permutation_disribution}
	\begin{algorithmic}
		\STATE Set $\tau(0) \in \mathcal{O}(P)$ to be the trivial equivalence relation with only one equivalence class $P$.
		\FOR{$t$ from $0$ to $1$ continuously}
		\WHILE{$\exists Y \subset X$ such that inequality~(\ref{eq:bounds}) is tight}
			\STATE Refine $\tau(t)$ be separating $X$ into $Y$ and $X \setminus Y$ appropriately.
		\ENDWHILE
		\ENDFOR
		\RETURN $\tau: [0,1] \ra \mathcal{O}(P)$.
	\end{algorithmic}
\end{algorithm}

This completes the proof of Lemma~\ref{lemma:diffusion_subgradient}.

%% file: main.bbl
\newcommand{\etalchar}[1]{$^{#1}$}
\begin{thebibliography}{GLPN93}

\bibitem[ABS10]{arora2010subexponential}
Sanjeev Arora, Boaz Barak, and David Steurer.
\newblock Subexponential algorithms for unique games and related problems.
\newblock In {\em Foundations of Computer Science (FOCS), 2010 51st Annual IEEE
  Symposium on}, pages 563--572. IEEE, 2010.

\bibitem[Alo86]{alon1986eigenvalues}
Noga Alon.
\newblock Eigenvalues and expanders.
\newblock {\em Combinatorica}, 6(2):83--96, 1986.

\bibitem[AM85]{alon1985lambda1}
Noga Alon and Vitali~D Milman.
\newblock $\lambda$1, isoperimetric inequalities for graphs, and
  superconcentrators.
\newblock {\em Journal of Combinatorial Theory, Series B}, 38(1):73--88, 1985.

\bibitem[Bau12]{bauer2012normalized}
Frank Bauer.
\newblock Normalized graph laplacians for directed graphs.
\newblock {\em Linear Algebra and its Applications}, 436(11):4193--4222, 2012.

\bibitem[Chu93]{chung1993laplacian}
F~Chung.
\newblock The laplacian of a hypergraph.
\newblock {\em Expanding graphs (DIMACS series)}, pages 21--36, 1993.

\bibitem[Chu97]{chung1997spectral}
Fan~RK Chung.
\newblock {\em Spectral graph theory}, volume~92.
\newblock American Mathematical Soc., 1997.

\bibitem[Chu05]{chung2005laplacians}
Fan Chung.
\newblock Laplacians and the cheeger inequality for directed graphs.
\newblock {\em Annals of Combinatorics}, 9(1):1--19, 2005.

\bibitem[CLTZ16]{chan2016spectral}
TH~Chan, Anand Louis, Zhihao~Gavin Tang, and Chenzi Zhang.
\newblock Spectral properties of hypergraph laplacian and approximation
  algorithms.
\newblock {\em arXiv preprint arXiv:1605.01483}, 2016.

\bibitem[FW95]{friedman1995second}
Joel Friedman and Avi Wigderson.
\newblock On the second eigenvalue of hypergraphs.
\newblock {\em Combinatorica}, 15(1):43--65, 1995.

\bibitem[GLPN93]{gallo1993directed}
Giorgio Gallo, Giustino Longo, Stefano Pallottino, and Sang Nguyen.
\newblock Directed hypergraphs and applications.
\newblock {\em Discrete applied mathematics}, 42(2):177--201, 1993.

\bibitem[HLW06]{hoory2006expander}
Shlomo Hoory, Nathan Linial, and Avi Wigderson.
\newblock Expander graphs and their applications.
\newblock {\em Bulletin of the American Mathematical Society}, 43(4):439--561,
  2006.

\bibitem[HSJR13]{hein2013total}
Matthias Hein, Simon Setzer, Leonardo Jost, and Syama~Sundar Rangapuram.
\newblock The total variation on hypergraphs-learning on hypergraphs revisited.
\newblock In {\em Advances in Neural Information Processing Systems}, pages
  2427--2435, 2013.

\bibitem[KLL{\etalchar{+}}13]{kwok2013improved}
Tsz~Chiu Kwok, Lap~Chi Lau, Yin~Tat Lee, Shayan Oveis~Gharan, and Luca
  Trevisan.
\newblock Improved cheeger's inequality: analysis of spectral partitioning
  algorithms through higher order spectral gap.
\newblock In {\em Proceedings of the forty-fifth annual ACM symposium on Theory
  of computing}, pages 11--20. ACM, 2013.

\bibitem[KLL15]{kwok2015improved}
Tsz~Chiu Kwok, Lap~Chi Lau, and Yin~Tat Lee.
\newblock Improved cheeger's inequality and analysis of local graph
  partitioning using vertex expansion and expansion profile.
\newblock {\em arXiv preprint arXiv:1504.00686}, 2015.

\bibitem[KVV04]{jacm/KannanVV04}
Ravi Kannan, Santosh Vempala, and Adrian Vetta.
\newblock On clusterings: Good, bad and spectral.
\newblock {\em J. {ACM}}, 51(3):497--515, 2004.

\bibitem[LGT14]{lee2014multiway}
James~R Lee, Shayan~Oveis Gharan, and Luca Trevisan.
\newblock Multiway spectral partitioning and higher-order cheeger inequalities.
\newblock {\em Journal of the ACM (JACM)}, 61(6):37, 2014.

\bibitem[LM14]{louis2014approximation}
Anand Louis and Konstantin Makarychev.
\newblock Approximation algorithm for sparsest k-partitioning.
\newblock In {\em Proceedings of the Twenty-Fifth Annual ACM-SIAM Symposium on
  Discrete Algorithms}, pages 1244--1255. SIAM, 2014.

\bibitem[Lou15]{louis2015hypergraph}
Anand Louis.
\newblock Hypergraph markov operators, eigenvalues and approximation
  algorithms.
\newblock In {\em Proceedings of the Forty-Seventh Annual ACM on Symposium on
  Theory of Computing}, pages 713--722. ACM, 2015.

\bibitem[LRTV11]{louis2011algorithmic}
Anand Louis, Prasad Raghavendra, Prasad Tetali, and Santosh Vempala.
\newblock Algorithmic extensions of {C}heeger's inequality to higher
  eigenvalues and partitions.
\newblock In {\em Approximation, Randomization, and Combinatorial Optimization.
  Algorithms and Techniques}, pages 315--326. Springer, 2011.

\bibitem[LRTV12]{louis2012many}
Anand Louis, Prasad Raghavendra, Prasad Tetali, and Santosh Vempala.
\newblock Many sparse cuts via higher eigenvalues.
\newblock In {\em Proceedings of the forty-fourth annual ACM symposium on
  Theory of computing}, pages 1131--1140. ACM, 2012.

\bibitem[LZ12]{li2012digraph}
Yanhua Li and Zhi-Li Zhang.
\newblock Digraph laplacian and the degree of asymmetry.
\newblock {\em Internet Mathematics}, 8(4):381--401, 2012.

\bibitem[MMV15]{colt/MakarychevMV15}
Konstantin Makarychev, Yury Makarychev, and Aravindan Vijayaraghavan.
\newblock Correlation clustering with noisy partial information.
\newblock In {\em {COLT}}, volume~40 of {\em {JMLR} Workshop and Conference
  Proceedings}, pages 1321--1342. JMLR.org, 2015.

\bibitem[PSZ15]{PengSZ15}
Richard Peng, He~Sun, and Luca Zanetti.
\newblock Partitioning well-clustered graphs: Spectral clustering works!
\newblock In {\em {COLT}}, volume~40 of {\em {JMLR} Workshop and Conference
  Proceedings}, pages 1423--1455. JMLR.org, 2015.

\bibitem[Rod09]{rodriguez2009laplacian}
JA~Rodr{\'\i}guez.
\newblock Laplacian eigenvalues and partition problems in hypergraphs.
\newblock {\em Applied Mathematics Letters}, 22(6):916--921, 2009.

\bibitem[SKR85]{Shor1985}
N.~Z. Shor, Krzysztof~C. Kiwiel, and Andrzej Ruszcay\`{n}ski.
\newblock {\em Minimization Methods for Non-differentiable Functions}.
\newblock Springer-Verlag New York, Inc., New York, NY, USA, 1985.

\bibitem[XQ16]{xie2016spectral}
Jinshan Xie and Liqun Qi.
\newblock Spectral directed hypergraph theory via tensors.
\newblock {\em Linear and Multilinear Algebra}, pages 1--15, 2016.

\bibitem[Yos16]{yoshida2016nonlinear}
Yuichi Yoshida.
\newblock Nonlinear laplacian for digraphs and its applications to network
  analysis.
\newblock In {\em Proceedings of the Ninth ACM International Conference on Web
  Search and Data Mining}, pages 483--492. ACM, 2016.

\bibitem[ZBL{\etalchar{+}}04]{zhou2004learning}
Dengyong Zhou, Olivier Bousquet, Thomas~Navin Lal, Jason Weston, and Bernhard
  Sch{\"o}lkopf.
\newblock Learning with local and global consistency.
\newblock {\em Advances in neural information processing systems},
  16(16):321--328, 2004.

\bibitem[ZGL{\etalchar{+}}03]{zhu2003semi}
Xiaojin Zhu, Zoubin Ghahramani, John Lafferty, et~al.
\newblock Semi-supervised learning using gaussian fields and harmonic
  functions.
\newblock In {\em ICML}, volume~3, pages 912--919, 2003.

\bibitem[ZHS06]{zhou2006learning}
Dengyong Zhou, Jiayuan Huang, and Bernhard Sch{\"o}lkopf.
\newblock Learning with hypergraphs: Clustering, classification, and embedding.
\newblock In {\em Advances in neural information processing systems}, pages
  1601--1608, 2006.

\bibitem[ZHTC17]{ZhangHTC17}
Chenzi Zhang, Shuguang Hu, Zhihao~Gavin Tang, and T.{-}H.~Hubert Chan.
\newblock Re-revisiting learning on hypergraphs: Confidence interval and
  subgradient method.
\newblock In {\em {ICML}}, volume~70 of {\em Proceedings of Machine Learning
  Research}, pages 4026--4034. {PMLR}, 2017.

\end{thebibliography}
